\documentclass[10pt]{article}
\usepackage[utf8]{inputenc}
\usepackage[margin=1in]{geometry}
\usepackage{rotating}
\usepackage{booktabs}
\usepackage{array}
\usepackage{amsmath, amsfonts}
\usepackage{color}
\usepackage[bookmarks=false]{hyperref}
\usepackage{bbm}
\usepackage{algorithm}
\usepackage[noend]{algpseudocode}
\usepackage{subcaption}
\usepackage{comment}
\usepackage{amssymb}
\usepackage{tikz}

\usepackage{mathtools}
\mathtoolsset{showonlyrefs}

\usepackage[
backend=bibtex,
style=ieee,
citestyle=numeric-comp,
sortcites,
sorting=none,
giveninits=true,
natbib,
hyperref,
maxbibnames=99,
doi=true,isbn=false,url=true,eprint=false
]{biblatex}

\usepackage{amsthm}
\usepackage[normalem]{ulem}
\usepackage{soul}

\newtheorem{theorem}{Theorem}

\newtheorem{definition}{Definition}

\def\R{\mathbb{R}}
\def\C{\mathbb{C}}
\def\tilh{\tilde{h}}

\def\tilf{\tilde{f}}
\def\F{\mathcal{F}}

\newcommand*\lri{\vcenter{\hbox{\includegraphics[width=0.75em]{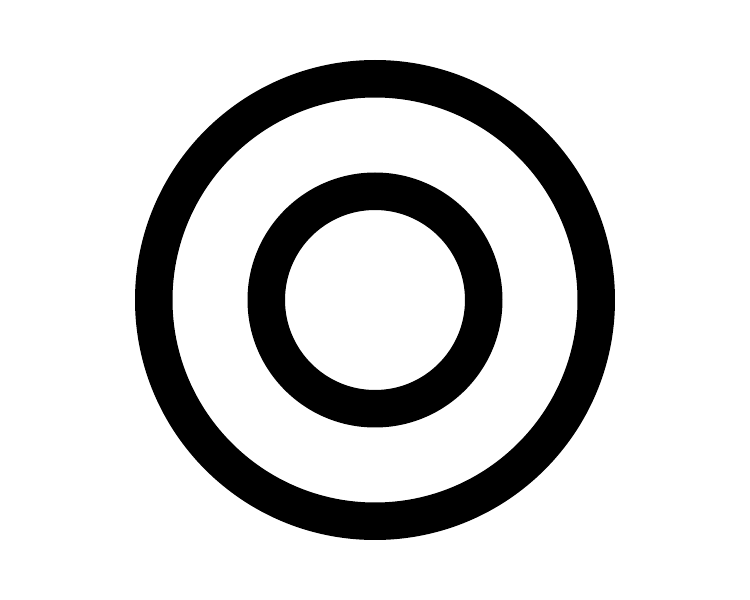}}}}

\usepackage{amsmath}

\DeclareMathOperator*{\argmin}{arg\,min}
\DeclareMathOperator*{\fft}{FFT}
\DeclareMathOperator*{\ifft}{iFFT}

\title{Ring deconvolution microscopy: exploiting symmetry for efficient spatially varying aberration correction}
\author{Amit Kohli*, \and Anastasios N.  Angelopoulos*, \and David McAllister, \and Esther Whang, \and Sixian You, \and Kyrollos Yanny,\and Federico M. Gasparoli, \and Bo-Jui Chang, \and Reto Fiolka, \and Laura Waller}

\begin{document}

\maketitle
\def\thefootnote{*}\footnotetext{These authors contributed equally to this work. Code available \href{https://github.com/apsk14/rdmpy}{here}.}\def\thefootnote{\arabic{footnote}}

\begin{abstract}
 The most ubiquitous form of computational aberration correction for microscopy is deconvolution. However, deconvolution relies on the assumption that the point spread function is the same across the entire field-of-view.
This assumption is often inadequate, but space-variant deblurring techniques generally require impractical amounts of calibration and computation.
 We present a new imaging pipeline that leverages symmetry to provide simple and fast spatially-varying aberration correction. Our \emph{ring deconvolution microscopy} (RDM) method leverages the rotational symmetry of most microscopes and cameras, and naturally extends to \emph{sheet deconvolution} in the case of lateral symmetry.
 We formally derive theory and algorithms for image recovery and additionally propose a neural network based on Seidel coefficients as a fast alternative, as well as extension of RDM to blind deconvolution. We demonstrate significant improvements in speed and image quality as compared to standard deconvolution and existing spatially-varying deconvolution across a diverse range of microscope modalities, including miniature microscopy, multicolor fluorescence microscopy, point-scanning multimode fiber micro-endoscopy, and light-sheet fluorescence microscopy.
 Our approach enables near-isotropic, subcellular resolution in each of these applications. 
   
\end{abstract}

\begin{figure}[t]
	\centering
	\includegraphics[width=\linewidth]{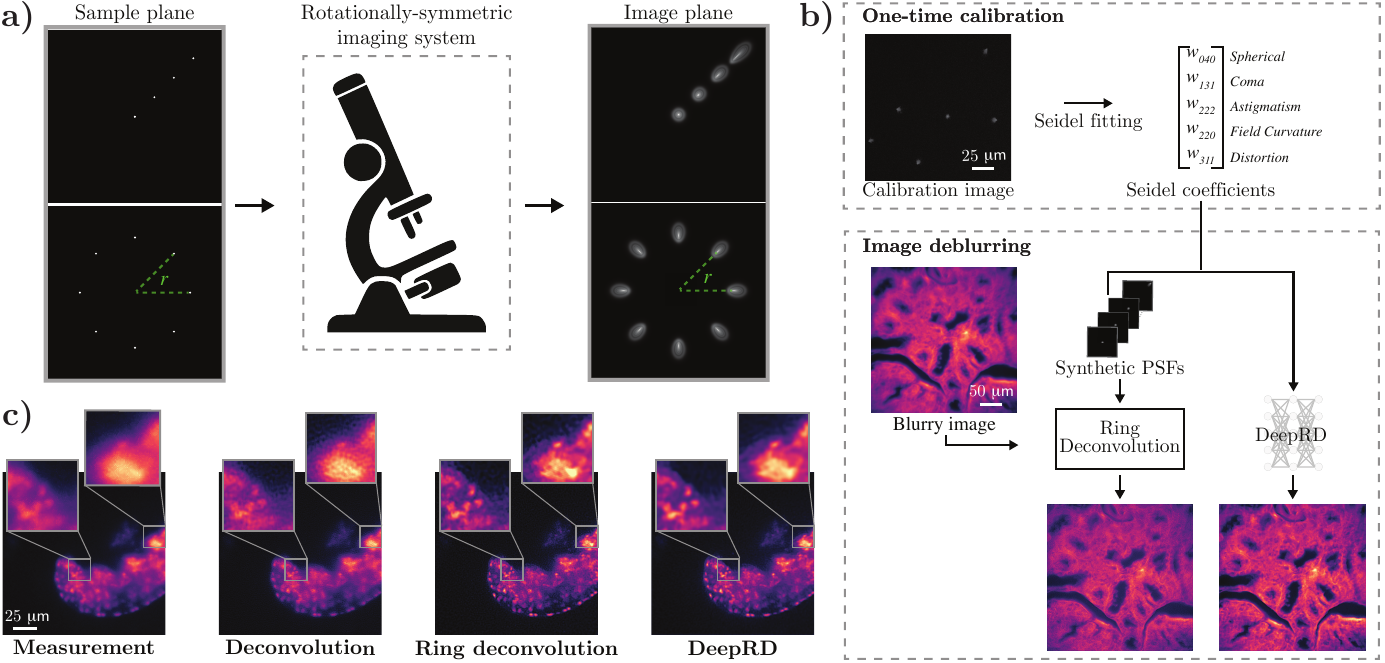}

	\caption{\textbf{Ring Deconvolution Microscopy (RDM).} \textbf{a)} Point sources at the sample plane (left) are imaged (right) to point spread functions (PSFs) with a rotationally-symmetric imaging system. The PSFs are linear revolution-invariant (LRI)---they vary with distance from the center of the field-of-view (FoV) (top row), but maintain the same shape at a fixed radius $r$, just revolved around the center (bottom row). \textbf{b)} The RDM pipeline. A one-time calibration procedure (top) captures a single image of randomly-placed point sources (e.g., fluorescent beads), then fits the primary Seidel coefficients (see Sec.~\ref{subsec:seidel}). Next, we either use the Seidel coefficients to generate a radial line of synthetic PSFs, if using ring deconvolution, or we feed the coefficients directly into Deep Ring Deconvolution (DeepRD). After calibration, we can deblur images (bottom) using either ring deconvolution or DeepRD. \textbf{c}) Experimental deblurring of live tardigrade samples imaged with the UCLA Miniscope~\cite{aharoni}. From left to right: measurement, standard deconvolution, ring deconvolution, and DeepRD. Ring deconvolution and DeepRD consistently outperform deconvolution.}
	\label{fig:1}
\end{figure}

\section{Introduction}
   
Much of optical engineering is focused on reducing aberrations by adding additional corrective optical elements to an imaging system---consider a microscope objective, comprised of numerous lenses stacked in a housing. Such designs allow for high-performance imaging, but incur added cost, weight, and complication. 
Even with large and expensive lens stacks, it is difficult and, in some cases, impossible to correct all aberrations across a large area, and so aberrations are often what limit the usable field-of-view (FoV) of a system. Furthermore, some systems cannot accommodate any aberration-correction optics; for example, additional elements may not fit in miniaturized microscopes (e.g. Miniscope~\cite{aharoni}) and are prohibitively expensive for large-aperture telescopes~\cite{stahl2009preliminary}.

Faced with a poorly-corrected imaging system, the modern microscopist instead turns to computational aberration correction, where the burden is shifted onto computer algorithms applied after the image is acquired. The most commonly used correction technique, image deconvolution, captures a calibration image of a small point-like source, known as the point spread function (PSF), in order to characterize the aberrations. The PSF can then be used to computationally deconvolve any image taken with the system via simple and fast algorithms, to yield a deblurred result. A main limitation of this approach is that it assumes that the system's PSF does not vary spatially (i.e. the system is linear space-invariant, LSI). This assumption is usually only true near the center of the FoV, and optical designers often artificially sacrifice part of the system's FoV in order to maintain space invariance. 

To go beyond space-invariant limitations, a large community effort has gone towards heuristic forms of spatially-varying 'deconvolution', wherein one measures PSFs at multiple points within the FoV and uses them to correct the image. Such heuristics include: assuming each region of an image is locally LSI~\cite{section-Trussell}, adaptively splitting the FoV by first quantifying the degree of space-variance~\cite{quant-Lin, quant-lohmann}, interpolating PSFs~\cite{interp-denis, interp-Nagy}, decomposing the PSF into space-invariant orthogonal modes~\cite{modes-denis, modes-flicker, modes-Hong, modes-miraut, modes-popkin, modes-Sroubek, modes-todlauer, Turcotte:20}, and doing the same in Fourier space~\cite{modes-fourier-Lee}. These heuristics can approach rigorous recovery as the number of PSFs collected grows, possibly into the hundreds of thousands; however, the tradeoff in terms of the complexity of calibration and computation quickly becomes intractable. For example, in patch-wise deconvolution, the FoV is divided into patches, each of which gets deconvolved by a PSF measured at its center. Maximum accuracy is achieved when the patch size is reduced to a single pixel, but then a megapixel image would require a million PSF measurements 
and a computation time of hundreds of hours to deblur (see Fig. \ref{fig:2}c).

Another emerging modality is deep deblurring~\cite{Yanny-multiwiener, ren2018deep, Deng:22, dong2020deep, weigert2018content, UNet_Deblur, li2022incorporating}, in which varying amounts of system information are incorporated into a deep neural network.
On one end of the spectrum, networks that are primarily data-driven struggle with extrapolation beyond the training data, and tend to reproduce whatever biases existed therein, a particularly relevant point since many of them are trained on simulated data. On the other end, networks that incorporate physical information, such as calibrated PSFs, may have better generalization properties but suffer the same accuracy/efficiency tradeoff as patch-based methods. 
For these reasons, spatially-varying deblurring has not become commonplace among practitioners.
Thus, there is a need for spatially-varying deblurring methods that are effective, efficient, and robust.

Here, we propose a spatially-varying method that requires only a single calibration image and has reasonable compute time (tens of seconds on a megapixel image), while offering rigorous deblurring for imaging systems that are symmetric in some way. We focus on \emph{rotationally-symmetric} systems---i.e., systems that are symmetric about their optical axis---but also show an example with the lateral symmetry present in light-sheet microscopy.
Rotational symmetry occurs in many imaging systems by design, and a significant portion of optical theory is developed under this assumption. While some existing deblurring techniques have leveraged rotational symmetry, they are approximate and restricted to a specific subset of radially-varying blurs: those due to camera zoom~\cite{Webster_motion_blur,boracchi2008deblurring}, the specific case of a parabolic mirror~\cite{luan2018suppress}, and an approximate scheme only for blurs from a single lens by applying deconvolution to four concentric regions~\cite{zhang_radial_blur, zhang_restoration, zhang_analysis, Zhang_ringing} . Other work does the same for DSLR cameras and also requires three image channels from an RGB camera~\cite{inproceedings} . In contrast, what we propose applies to \textit{any} rotationally-symmetric imaging system, can incorporate more complex PSFs---even if they cannot be theoretically derived, makes no approximations (e.g., isoplanatic regions) in the image formation model, and can easily extend to other symmetries.

Our Ring Deconvolution Microscopy (RDM) technique achieves both accuracy and efficiency.
RDM models image formation for rotationally-symmetric imaging systems rigorously, allowing for accurate deblurring while still remaining practical, both computationally and in terms of calibration. The first step in RDM is a simple, single-shot calibration scheme, in which the system's primary Seidel aberration coefficients are estimated from a single image of randomly distributed point sources. These coefficients quantify the severity of spatial variance and provide the necessary system information for the second step, deblurring.
Within this second step we propose two alternative image deblurring algorithms. 
The first (our main algorithm) is ring deconvolution, which uses a new and rigorous theory for rotationally-symmetric imaging to deblur the image at all points in the FoV, with only order $N^3\log(N)$ ($N$ is the image side length) compute time, as compared to $N^4$ for full spatially-varying deblurring. For even faster computation (but without theoretical guarantees), we further propose an alternate neural network-based algorithm called Deep Ring Deconvolution (DeepRD), which constrains learning with physical knowledge provided by the system's Seidel aberration coefficients. If the system is not spatially-varying or minimally so, RDM still offers an improvement over standard deconvolution by instead deconvolving with a synthetic PSF generated by the Seidel coefficients (see Sec.~\ref{sec:results:extensions}).

Although ring deconvolution is specific to systems exhibiting rotational symmetry, our theory can be easily adapted to exploit other forms of symmetry. As an example, we derive an analogous form of deconvolution for situations where the blur varies laterally (along one Cartesian axis), which we term \emph{sheet deconvolution}. This is the case in light-sheet microscopy, where the light-sheet illumination causes space-varying blur in the direction perpendicular to the imaging plane. We show experimental results for sheet deconvolution on light-sheet fluorescence microscopy data in Sec.~\ref{methods:sheet}, in which we effectively increase the FoV by about twofold through our deblurring process. 

These algorithms outperform existing methods, approaching subcellular, isotropic resolution across the FoV. We demonstrate this on four diverse microscope modalities: miniature microscopy, multicolor fluorescence microscopy, multimode fiber micro-endoscopy, and light-sheet fluorescence microscopy. Each of these modalities contains different characteristics and imaging mechanisms that are representative of a wide range of imaging systems, thereby forming a comprehensive basis to demonstrate the wide applicability and practical relevance of our methods.
See Fig.~\ref{fig:1} for a summary of the RDM pipeline along with an example result on images of live tardigrades. An open source implementation of RDM and its extensions can be found \href{https://github.com/apsk14/rdmpy}{here}.


\section{Results}
\label{sec:results}

Before we display our experimental results, we briefly outline the RDM pipeline from calibration to deblurring. Further details about RDM can be found in Sec.~\ref{sec:methods}.

\subsection{Ring Deconvolution Microscopy Pipeline}



The first step in our RDM pipeline is calibration; we measure the system's response to a point source, i.e., its PSF. The PSFs of a space-varying system will vary across the FoV, and many PSF measurements may be required to fully characterize the system. However, space-varying systems that are rotationally symmetric require fewer measurements for system characterization, which we exploit here. Intuitively, PSFs that are the same distance from the center of the FoV all have the same shape, just rotated at different angles  (see Figure~\ref{fig:1}a) because of the symmetry. We call this property of the PSFs \emph{linear revolution-invariance} (LRI), and denote it mathematically as
\begin{equation}
        \tilh(\rho, \phi; r, \theta) =  \tilh(\rho, \phi - \theta; r, 0),
\label{eq:lri}
\end{equation}
\noindent where $\tilh(\rho, \phi; r, \theta)$ is the (spatially-varying) PSF in polar coordinates from a point source at location $(r, \theta)$. Note that the shape of the PSF itself is not necessarily rotationally symmetric. LRI greatly improves complexity of the calibration procedure, in that we need only measure the PSF at any one point along the circle for each radius $r$ from the optical center. 

In practice, directly measuring the PSF at every radius $r$ is impractical. Instead, under LRI, there is a simple and effective method for simultaneously estimating these PSFs from a single image, without any motion stage required. The method is described in Section~\ref{subsec:seidel} and is the first step of the RDM pipeline. It works by estimating the Seidel aberration coefficients from an image of point sources randomly scattered across the FoV (see Fig.~\ref{fig:1}b). Seidel coefficients are a natural choice for LRI systems, since Seidel polynomials are explicit functions of the field position and thus use the same coefficients regardless of the position in the FoV~\cite{born2013principles}. Thus, they are easier to estimate than Zernike coefficients, which are different at each field position~\cite{Zheng:13, Shao:19}. Our experiments demonstrate that the 5 primary (4th order) Seidel coefficients (sphere, coma, astigmatism, field curvature, and distortion) suffice to characterize spatially-varying LRI systems. A discussion of higher-order aberrations can be found in Appendix~\ref{appendix:higher-order}.


After calibration, the next step is to use the estimated Seidel coefficients to deblur the measured image. 
LRI systems, much like LSI systems, allow for computationally-efficient models. For LSI systems, we leverage two axes of space-invariance to model the blur as a 2D convolution, and for LRI systems we can leverage the revolution-invariant angular axis, leading to a blur computation that is a sum of 1D convolutions.
Section~\ref{subsec:ring} provides a description of the forward and inverse algorithms, which we call \emph{ring convolution} and \emph{ring deconvolution}, respectively. 
Therein is also included a theoretical proof of ring convolution's exactness under rotational symmetry.
We briefly describe ring deconvolution here, along with an alternate method using deep learning, termed \emph{DeepRD}.
These methods form the second step in the RDM pipeline (Fig.~\ref{fig:1}b).
\begin{enumerate}

\item \textbf{Ring deconvolution.}  We derive an optimal algorithm for reconstructing the underlying sample from a blurry image given the PSF at each radius of the FoV of a rotationally-symmetric imaging system. This is our main algorithm.

\item \textbf{Deep Ring Deconvolution (DeepRD).} Although ring deconvolution is significantly faster to compute than a full patch-based spatially-varying deblur technique, it may still be relatively slow (on the order of a few minutes) for very large image sizes (e.g., beyond $1024 \times 1024)$ or video data. Deep learning enables a faster (but approximate) version of ring deconvolution called DeepRD. 
As input, it takes a blurry image and a list of the 5 primary Seidel coefficients.
DeepRD is trained on a dataset of natural images that are synthetically blurred using ring convolution.
\end{enumerate}

\begin{figure}[p]
 \centering
 \includegraphics[width=\linewidth]{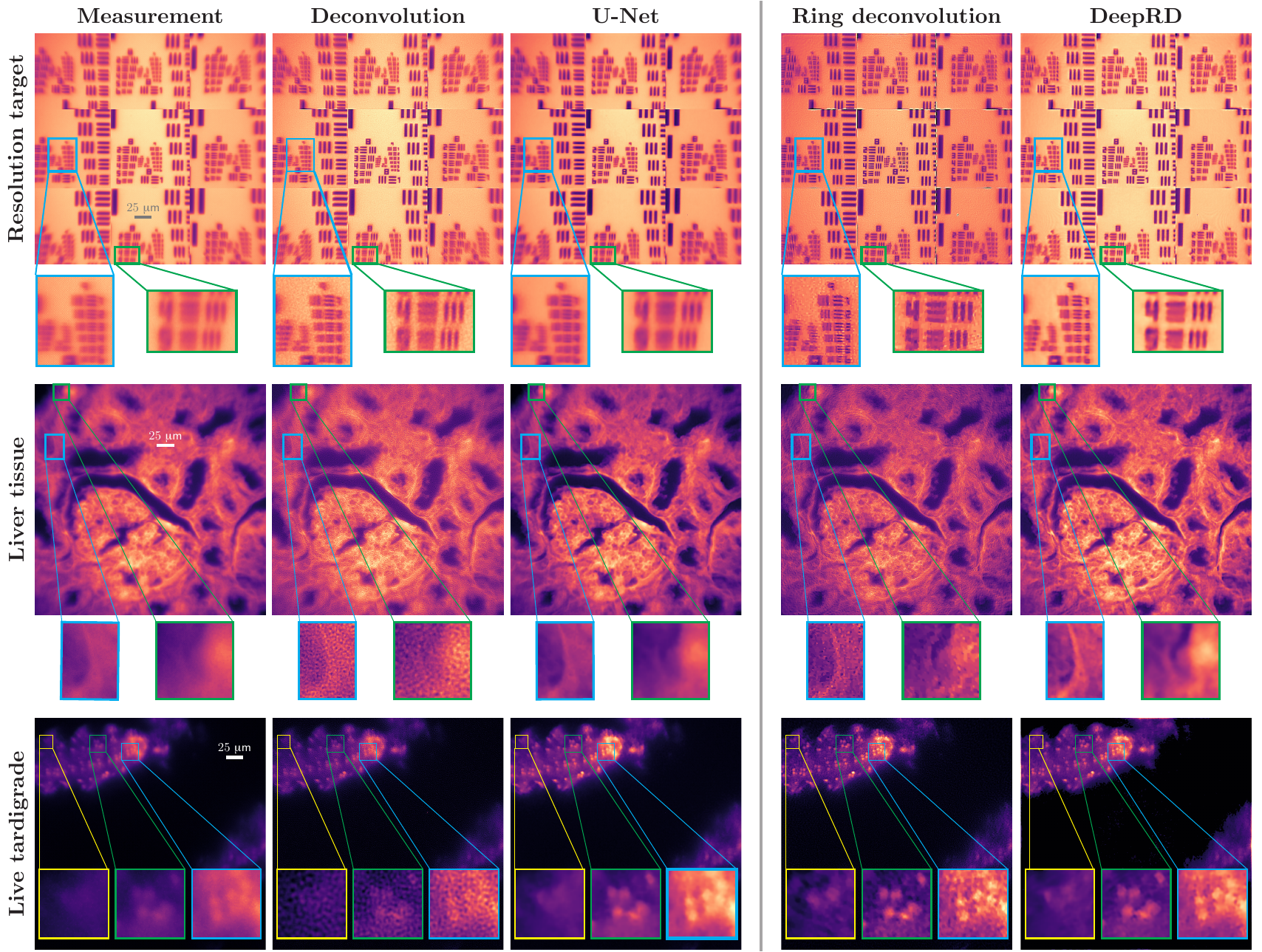}
\caption{\textbf{Ring Deconvolution Microscopy (RDM) for the UCLA Miniscope.} After calibrating the Miniscope with a single image of fluorescent beads (see Fig.~\ref{fig:1}b), we show results from several deblurring methods for comparison: standard deconvolution (using  a single measured PSF), a U-Net trained on our spatially-varying blur dataset, ring deconvolution and DeepRD. Deconvolution assumes space-invariance while the remaining methods are designed to handle spatially-varying aberrations. In the first row, ring deconvolution and DeepRD clearly resolve resolution target elements near the edges of the FoV, which are not as well resolved by the two other methods. Zoom-ins show RDM resolves up to element 6 of group 9 (blue inset) and element 5 of group 8 (green inset). Similar results along with corresponding insets are shown for the other samples.}
 \label{fig:3}
\end{figure}


\subsection{Experimental Results}
\subsubsection*{Miniature microscopy}

Our first experimental system is the UCLA Miniscope~\cite{aharoni}, a miniature microscope used largely for neuroimaging in freely-behaving animals. The first step of our experimental pipeline is to calibrate the imaging system by capturing a single image of randomly-placed fluorescent beads (obtained by smearing bead solution on a slide). We use it to fit Seidel coefficients, obtaining 0.85, 0.56, 0.25, 0.29, and 0 waves of spherical aberration, coma, astigmatism, field curvature and distortion, respectively. These numbers, while specific to our particular assembly of the Miniscope, are consistent with the aberration profile of a radial GRIN lens~\cite{Wang:90}, which is the objective lens used by the Miniscope. The fact that the off-axis coefficients (i.e. all of the primary coefficients except for spherical) are nonzero confirms that the system is indeed spatially-varying. For comparison with standard deconvolution deblurring, we also acquire the center PSF by imaging a single fluorescent bead isolated and centered in the FoV. The PSF is then denoised before its use in deconvolution (see~\ref{subsec:exp_deets}).

Having completed the one-time calibration for the Miniscope, we image a variety of different samples: a USAF resolution target, rabbit liver tissue, and live fluorescence-stained tardigrades. For each sample, we compare reconstructions from ring deconvolution and DeepRD, as well as standard deconvolution (using the PSF measured at the center of the FoV) and a baseline U-Net (see Fig.~\ref{fig:3}). 

The USAF resolution target was placed at 9 separate locations in the FoV, which allows us to observe the smallest features at each location. We stitch these 9 images together, using only the region of each constituent image that contains the high-resolution group, to form the measurement image shown in Fig.~\ref{fig:3} (first row). Ring deconvolution and DeepRD---having knowledge of the field-varying aberrations via the Seidel coefficients---give the most improvement near the edges and corners of the image. Standard deconvolution produces a noisy, low contrast result in those regions due to the mismatch of the center PSF with the edge PSFs. We also note that both learning methods (U-Net and DeepRD) offer the best denoising performance, a well-known property of neural networks \cite{vincent2008extracting}; however, this comes at the cost of inconsistent performance---both models perform worse on the resolution target than on the other samples.

For our second sample, rabbit liver tissue, our methods perform well, revealing features in the corners of the image, including the outlines of membranes, which are not clear in the raw measurement or other methods. Such clarity is critical for downstream tasks like image segmentation.

Finally, we fluorescently stain live semi-starved tardigrades with a DNA gel stain and capture a series of videos. We apply deblurring to each frame and display one such frame in the bottom row of Fig.~\ref{fig:3}; the full videos can be found \href{https://berkeley.box.com/s/d1o1901uv8ehxdf7kzdapyej1c7dt6bi}{here}. As with the previous experiments, ring deconvolution and DeepRD provide increased image contrast and detail in the corners of the frames compared to the other methods. In particular, the small, dot-like features within the tardigrade are better resolved.

 Further experimental details can be found in Sec.~\ref{subsec:exp_deets}. The baseline U-Net we use is inspired by the popular Content-Aware Image Restoration (CARE) model~\cite{weigert2018content}; we use a similar model structure but train both DeepRD and the baseline U-Net using a novel synthetic training strategy: applying ring convolutions using PSFs generated from random Seidel coefficients to natural images from Div2k~\cite{agustsson2017ntire}(see Sec.~\ref{subsec:deeprd}). The dataset is publicly available \href{https://berkeley.app.box.com/folder/252726027686?s=vv3g6avhrr9agijmlj3b1153oo7x9gao}{here}. Despite being trained on simulated data, both the U-Net and DeepRD exhibit reasonable deconvolution performance on experimental images, as shown in Fig.~\ref{fig:2}. Further, in Appendix~\ref{appendix:DeepRD}, we provide a demonstration that DeepRD is interpretable: it understands the individual effects of each aberration coefficient.
 
\subsubsection*{High NA multicolor fluorescence microscopy}

Our second experimental system to demonstrate RDM is a high magnification, high numerical aperture (NA) microscope. Such devices are critical to observing biological samples at subcellular resolution. However, as the NA increases, so do field-varying aberrations. RDM offers a pathway to utilize the level of magnification and NA needed for subcellular imaging whilst maintaining isotropic resolution over the entire FoV. Moreover, RDM does this efficiently over multiple fluorescence color channels, allowing for multicolor, subcellular-resolution imaging over the entire FoV. To demonstrate this, we image the fluorescently-labeled actin (green channel) and mitochondria (red channel) of Bovine Pulmonary Artery Endothelial (BPAE) cells with a 100$\times$ 1.4 NA objective. More details of the sample and imaging tools are found in~\ref{subsec:exp_deets}.

Again, we calibrate the system by capturing a single image of randomly scattered beads (see Fig.~\ref{fig:multicolor} a)). We choose to perform a separate calibration for each color channel, with different bead images corresponding to the different emission wavelengths; this strategy allows RDM to additionally correct chromatic aberrations. In practice, however, there is very little difference between the resulting estimated red and green Seidel coefficients.

After calibration, we image BPAE cells and process them with both RDM and standard deconvolution, for comparison. Figure~\ref{fig:multicolor} b) shows that RDM deblurs the raw images consistently over the entire FoV, including the corners of the image, while standard deconvolution becomes low contrast and noisy near the edges. In both examples, RDM is able to resolve subcellular features in the actin and mitochondria near the edges that are not visible in standard deconvolution. Such capability allows for larger FoVs to be used, lessening the burden of mechanically scanning and stitching together many smaller FoV images when the sample is large.

\begin{figure}[p]
	\centering
        \vspace{-1cm}
	\includegraphics[width=0.85\linewidth]{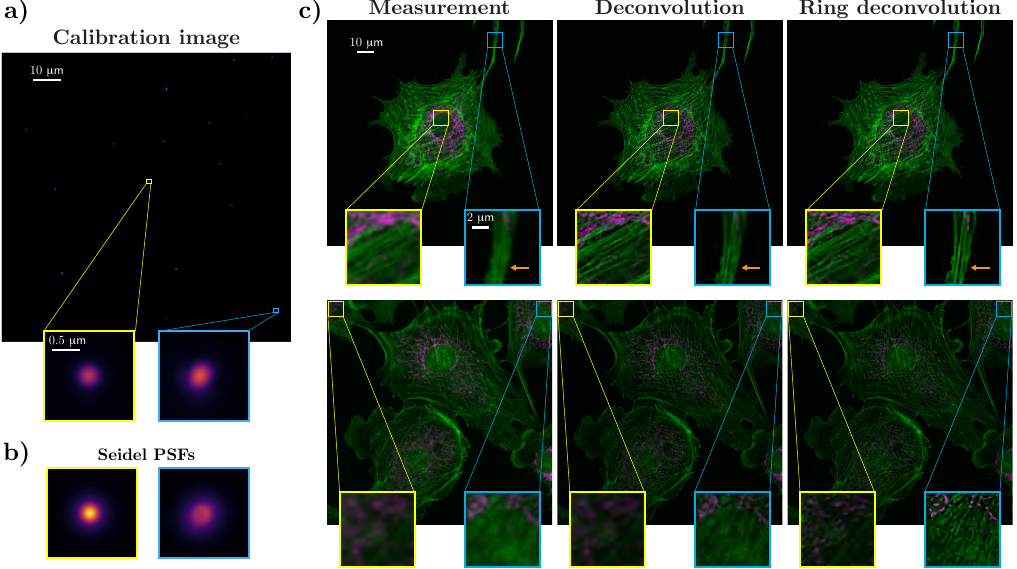}
	\caption{\textbf{RDM on multicolor fluorescence imaging.}
 \textbf{a)} Fluorescent beads imaged with a 100x 1.4 NA objective and \textbf{b)} corresponding Seidel-fitted PSFs demonstrate spatially-varying nature of the system. \textbf{c)} Two examples of Bovine Pulmonary Artery Endothelial (BPAE) cells processed by standard deconvolution and RDM. Deconvolution and RDM perform similarly in the center but RDM is better in the corner, revealing submicron features in the actin (orange arrow). RDM similarly resolves actin filaments and mitochondria where deconvolution does not.}
    \label{fig:multicolor}
\end{figure}

\begin{figure}[p]
	\centering
        \vspace{-1cm}
	\includegraphics[width=0.8\linewidth]{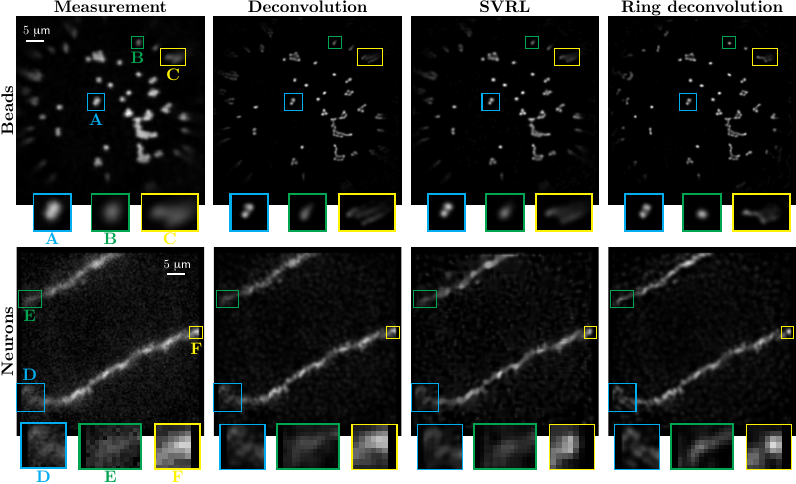}
	\caption{\textbf{RDM for point-scanning micro-endoscopy through a multimode fiber.} Comparison of deconvolution, spatially-varying Richardson-Lucy (SVRL), and RDM on images taken from a point-scannning multimode fiber micro-endoscope~\cite{Turcotte:20}. SVRL results are directly obtained from~\cite{Turcotte:20}. Top row: Sample with 1 $\mu m$ beads; all methods resolve the circularly-shaped beads in the center (blue inset). Away from the center, however, deconvolution and SVRL fail to completely remove the ellipiticity of the bead but RDM does (green inset). Moreover, in the corner only RDM can resolve bead clusters into their component beads (yellow inset). Bottom row: live Wistar rat neuron; RDM can clearly resolve structures (blue inset), sharpen the neuron spine (green inset), and reveal a point-like structure (yellow inset) near the edges where the other methods cannot.}
    \label{fig:mmf}
\end{figure}

\subsubsection*{Micro-endoscopy through a multimode fiber}

Point-scanning micro-endoscopy through a multimode fiber~\cite{Turcotte:20,vasquez2018subcellular} is a powerful technique for deep \emph{in vivo} imaging at subcellular resolution, with applications in the brain and other sensitive organs, where minimal tissue damage is required. However, due to the extreme constraints imposed in the design of the fiber, its resolution capabilities degrade rapidly and severely away from the center of the image, resulting in a small usable FoV (see the top row of Fig.~\ref{fig:mmf}). The spatially-varying images from such a system have been heuristically deblurred in~\cite{Turcotte:20}, but RDM, with its rigorous formulation, offers improvement in performance with far less calibration. 

To verify this, we process images of beads and live rat neurons from a previous paper~\cite{Turcotte:20} using RDM, their spatially-varying Richardson-Lucy (SVRL) algorithm, and standard deconvolution, for comparison. SVRL is similar to the modal decomposition work in~\cite{yanny2020miniscope3d, flicker2005anisoplanatic}. It uses fewer synthetic PSFs than RDM (~30 vs. ~120), so requires slightly less memory, but requires significantly more calibration images. Here, SVRL is calibrated by measuring a uniform grid of 21$\times$21 PSFs---441 images---whereas RDM is calibrated with a single image of a few randomly-placed beads. Despite the considerably lighter calibration and similar computational complexity, RDM provides an improvement in image quality over SVRL (see Fig.~\ref{fig:mmf}). In the corners of the bead image, we see that RDM is able to remove the aberration-induced ellipticity of the underlying circular beads and resolve clumps of beads, unlike the other methods. The same holds for the neuron images---in the corners, RDM is able to resolve distinctive subcellular features and tighten the spread of thin spines far better than the other methods. Additional experiments on this system and comparisons with more methods can be found in Appendix~\ref{appendix:additional-comps}.

\begin{figure}[p]
	\centering
        \vspace{-1cm}
	\includegraphics[width=0.9\linewidth]{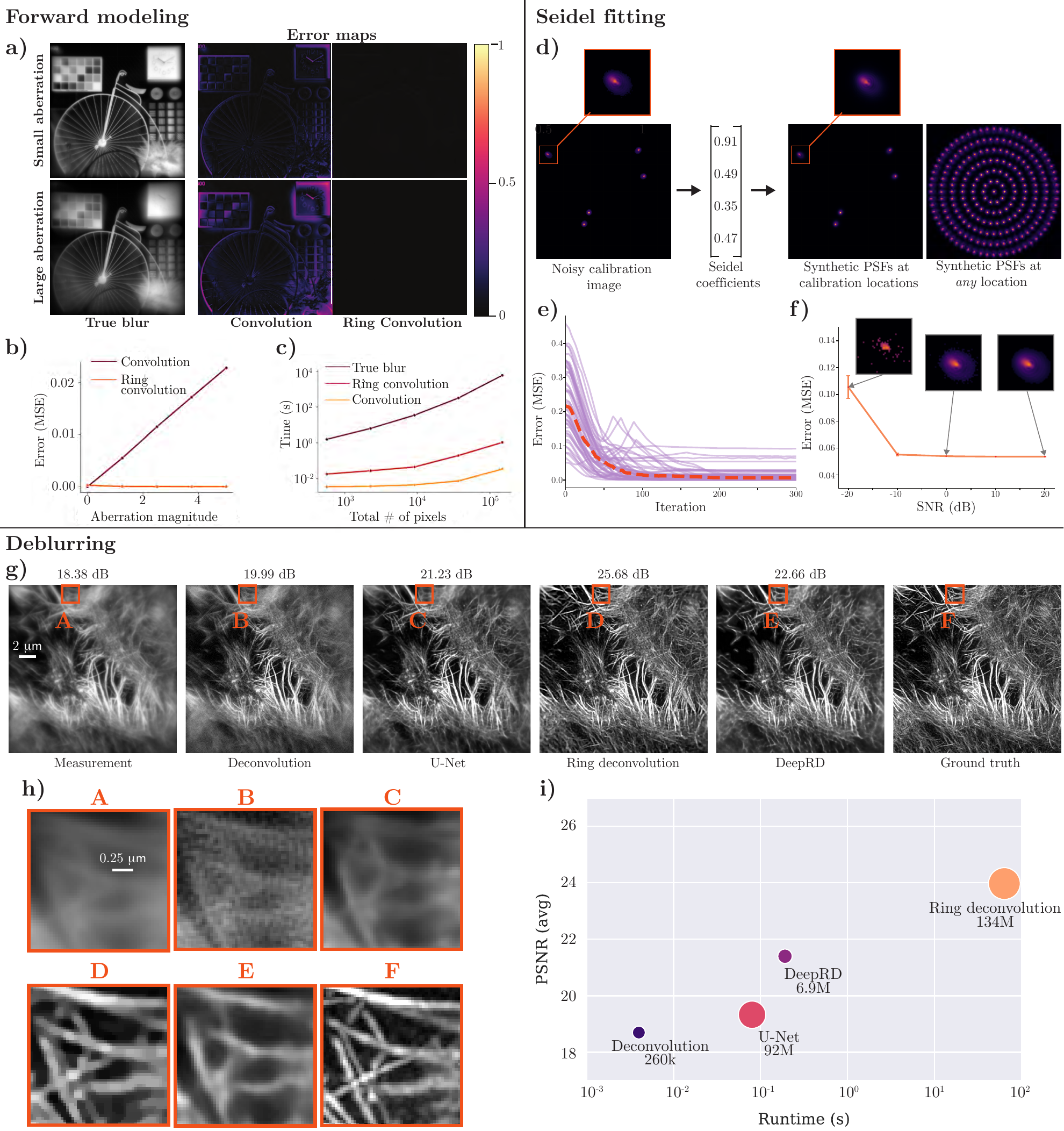}
	\caption{\textbf{Simulations to quantify RDM performance.}
 \textbf{a)} Forward model: ring convolution is compared with standard deconvolution by computing error maps that measure the absolute difference between the simulated blurs from each method and the \lq true blur\rq (produced by manually superimposing every PSF at every pixel). When off-axis aberrations are small (top row), both forward models are accurate. When aberrations are large (bottom row), convolution becomes significantly worse, yet ring convolution remains accurate. \textbf{b)} Forward model mean-squared error (MSE) as a function of off-axis aberration magnitude (for the image above). \textbf{c)} Runtime of each method as a function of size of the image (in pixels), averaged over n=50 trials. Computing the true blur quickly becomes intractable, but convolution and ring convolution remain relatively fast. \textbf{d)} A noisy image of randomly-placed PSFs is used to estimate the underlying Seidel coefficients, which can then be used to generate PSFs at any location. \textbf{e)} Seidel fitting error is plotted as a function of iteration in the optimization algorithm to demonstrate convergence. Each purple line is a different trial of n=50 trials with a randomly-sampled set of underlying Seidel coefficients. The red dashed line is the per-iteration median.  \textbf{f)} Average MSE of the fitted Seidel coefficients plotted as a function of SNR of the calibration image (with standard deviation shown by error bars) over n=50 random trials. Some example calibration PSFs are shown. Even under severe noise, the Seidel fit is still accurate. \textbf{g)} Deblurring results on noisy images from the CARE dataset, with PSNR values above each method. \textbf{h)} Zoom-ins of an off-axis patch in each deblurred image; ring deconvolution and DeepRD have the highest quality. \textbf{i)} Accuracy (PSNR) vs runtime of each method (averaged over n=28 true blurred images using unseen coefficients), with the number of model parameters (written below each circle) determining the size of the circle.  
    }
    \label{fig:2}
\end{figure}

In summary, we find that ring deconvolution consistently produces the best reconstructions among the methods we tested on our diverse panel of imaging systems. We believe the improvement arises due to the lack of approximations and heuristics in the method's derivation. Moreover, as compared to deep-learning-based methods, ring deconvolution undergoes no learning procedure, and thus does not transmit bias from the training data into future reconstructions. DeepRD performs similarly to ring deconvolution, but it is faster and less consistent. Both perform better than the U-Net and standard deconvolution.

\subsection{Simulation Results}

To quantify the performance of our methods, we conducted a series of simulations in which we have access to the unblurred ground-truth image. 
These simulations are done with Gaussian noise and are repeated with Poisson noise in Appendix~\ref{appendix:poisson}.

Our first step is to quantify the forward model: ring convolution. Simulating the blurring operation of an imaging system is a critical part of image deblurring. By mathematically specifying the forward model of the system, we can then invert it to obtain the desired deblurring algorithm. When a system is space invariant, the forward model is a simple convolution operation. For space-varying systems, however, we must account for the changes in the PSF across the FoV. The brute-force approach for doing so would superimpose weighted PSFs at each pixel in the image to compute the \lq true blur\rq, at the cost of long compute times. However, when the system is rotationally symmetric (i.e., varies only radially), ring convolution is an equivalent operation to the brute-force method, but runs much quicker---less than a second---even for image sizes upward of $512 \times 512$. To verify, we blur each test image using spatially-varying PSFs rendered from a randomly chosen set of Seidel coefficients. We treat the true blur as ground truth and compare the error maps for both standard LSI convolution and our ring convolution. As expected, standard convolution results in errors near the edges of the FoV, whereas ring convolution produces accurate blur across the entire image (Fig.~\ref{fig:2}a). We can quantify the aberration magnitude by calculating the norm of the off-axis Seidel coefficients; in Fig.~\ref{fig:2}b we see that the error for standard convolution increases approximately linearly with this aberration magnitude. In contrast, ring convolution continues to produce an accurate blur, independent of the strength of the aberrations. 
In Fig.~\ref{fig:2}c, we compare compute times for these forward models, showing that our ring convolution method is nearly four orders of magnitude faster than the other exact method (true blur) for a megapixel-sized image.  


Next, we verify our Seidel fitting method by quantifying the error in our estimated Seidel coefficients that were fitted from a single noisy image of randomly-scattered point sources. As detailed in Sec.~\ref{subsec:seidel}, the Seidel fitting procedure involves searching for the set of 5 primary Seidel coefficients that best fit the measured PSFs at their given positions in the calibration image. We simulate this process many times by randomly generating sets of Seidel coefficients, using them to produce a calibration image with additive Gaussian noise, and then estimating those coefficients only using the noisy calibration image (see Fig.~\ref{fig:2}d). Section~\ref{subsec:seidel} shows that this optimization problem is nonconvex, yet we find that our estimated Seidel coefficients almost always converge correctly. This is not entirely surprising, since the Seidel procedure is optimal in the maximum likelihood sense \cite{Southwell:77}. To quantify, we ran many Seidel fits over random coefficients, randomly located PSFs, and random draws of additive noise, and plotted the error between estimated and true coefficients in Fig.~\ref{fig:2}e. Due to nonconvexity, we see that not every case produces errors which converge to 0. However, two things provide assurance: 1) the median convergence approaches 0, meaning that a majority of optimizations will produce the optimal solution, and 2) even the runs which do not converge to the global minimum produce PSFs which are close enough to the true ones to provide good quality ring deconvolution. Finally, since the Seidel fitting procedure acts as a PSF denoiser, we tested the fit with varying amounts of noise (up to $-20$ dB SNR) and found that the fit only reduces performance slightly, even with severe additive noise (Fig.~\ref{fig:2}f). Additional simulations for higher-order Seidel coefficients can be found in Appendix~\ref{appendix:higher-order}.


After verifying that our forward model and Seidel fitting methods perform well, we use them as part of an inverse problem to deblur images and quantify the performance. We compare our methods (Ring deconvolution and DeepRD) with standard deconvolution and the baseline CARE-based U-Net. The learning methods are trained with images from the CARE and Div2k datasets~\cite{weigert2018content, agustsson2017ntire}, after synthetically blurring them with space-varying PSFs rendered from a random set of Seidel coefficients. It was computationally infeasible to generate the training set using the true blur technique, so instead we used our ring convolution to generate the blurred images for training. This should not pose a major issue since ring convolution is theoretically equivalent to the true blur case. To remove any doubt, we construct the test set using the true blur.
Each model is first pretrained on 80,500 data points from the Div2k dataset (800 base images blurred with 100 different Seidel coefficients) and then fine-tuned on a small 8400 data point subset of the CARE dataset (24 images, 350 different Seidel coefficients). After training, each method is tested on 
a test set of 28 unseen images from the CARE dataset, which are blurred using the true blur method and noised with additive Gaussian noise. The results of each method on one representative test image are shown in Fig.~\ref{fig:2}g,h, with the PSNR (in dB) listed above. Both DeepRD and ring deconvolution deblur better near the edges and corners of the image, where the PSF most deviates from the center PSF. Despite using ring convolution for the training set and the true blur for the test set, neither of the networks (U-Net and DeepRD) show signs of model mismatch. 

In Fig.~\ref{fig:2}i, we plot the average accuracy (PSNR) versus runtime across the full test dataset, with the size of the circle representing the number of parameters needed (memory footprint) for each method. Ring deconvolution provides the best reconstruction, but it is also the largest and slowest of the methods tested, taking about 60 seconds for a full image. However, if we were to try to deblur the image rigorously without using RDM, it would take hundreds of hours, despite being theoretically equivalent to ring deconvolution. Thus, ring deconvolution allows for relatively fast, accurate deblurring where it was once infeasible. DeepRD performs nearly as well as ring deconvolution and has the fewest parameters needed of all the space-varying techniques, allowing it to be fast and memory efficient. DeepRD is almost three orders of magnitude faster than RDM. The baseline U-Net and standard deconvolution PSNR values are significantly worse. 

\subsection{Sheet Deconvolution: Extension to Lateral Symmetry}

\label{subsec:ls}

\begin{figure}[t]
 \centering
\includegraphics[width=\linewidth]{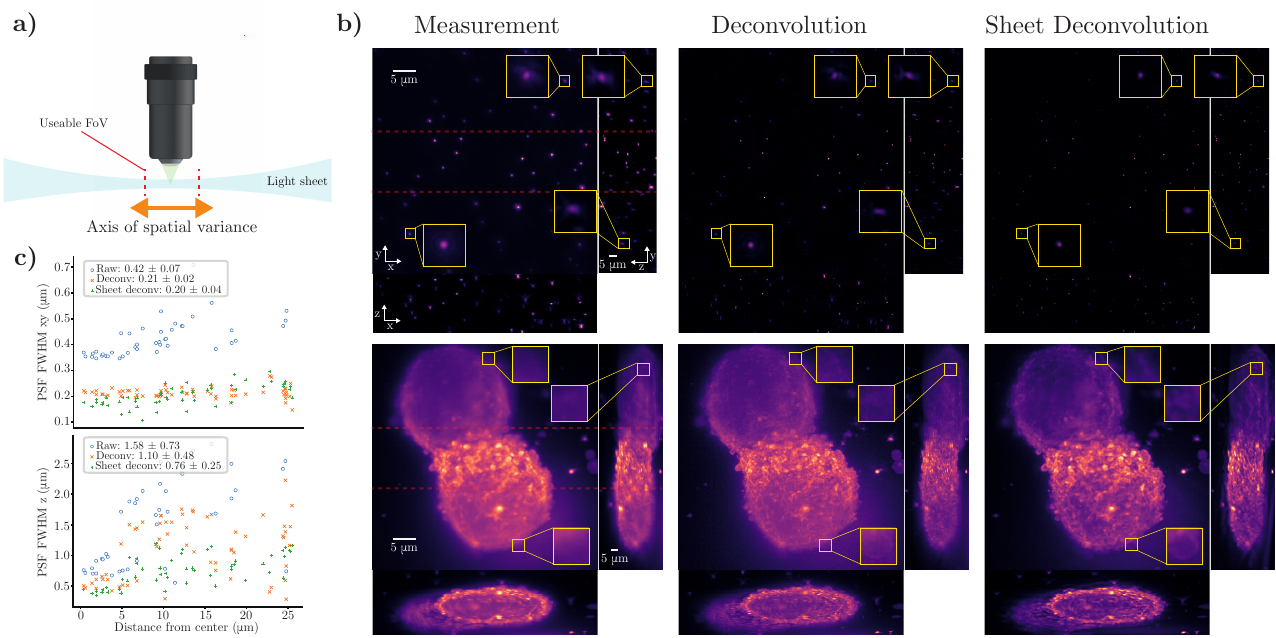}
 \caption{\textbf{Sheet deconvolution exploits lateral symmetry for light-sheet microscopy.} We run sheet deconvolution on samples imaged by a light-sheet fluorescence microscope with a detection NA of 1.1  and a light-sheet NA of 0.31. \textbf{a)} The light-sheet microscope illumination spreads along its propagation axis (orange arrow) such that only a center strip in the FoV (marked \lq useable FoV\rq) is optically sectioned and hence well resolved. \textbf{b)} Results comparing standard 3D deconvolution and sheet deconvolution on fluorescent beads (top row) and SU8686 cells (bottom row), with maximum intensity projections along each axis. Unlike standard deconvolution, sheet deconvolution resolves features nearly as well at the edges of the FoV as in the center (see insets). \textbf{c)} Plots of lateral and axial FWHM of reconstructed beads (from \textbf{b}) as a function of their distance from the center along the spatially-varying axis. For sheet deconvolution, the FWHM values grow less in lateral and axial diameter than that of standard deconvolution and differs by at most $500$ nm in diameter from center to the top/bottom of the FoV.}
 \label{fig:light-sheet}
\end{figure} 

Revolution-invariance is not the only form of symmetry found in practical, spatially-varying systems. For example, light-sheet fluorescence microscopy (LSFM) exhibits another type of symmetry, in which the light-sheet illumination causes spatial variance, but only along one axis (see Fig.~\ref{fig:light-sheet}a)). The thinnest section, or waist, of the light-sheet is focused in the center of the FoV and becomes thicker as a function of the distance from the center. Since the light-sheet is constant along the axis orthogonal to the focusing direction, this variance only occurs along one dimension. The resulting PSF grows in axial extent as it moves along one of the lateral dimensions, but stays constant in the other lateral dimension.

This symmetry can be used to obtain a fast spatially-varying forward model of the LSFM system. By employing the same mathematical strategy used in the ring convolution theorem, we can derive \emph{sheet convolution}, an exact forward model for LSFM (assuming a space-invariant objective lens). Just like ring convolution, sheet convolution leverages symmetry by convolving over the two dimensions that are spatially-invariant while integrating over the one that does vary. As a result, sheet convolution enjoys a similar improvement in computational performance as ring convolution; instead of the $O(N^6)$ scaling of the general three-dimensional spatially-varying forward model, it achieves a $O(N^4log(N))$ scaling. This improvement enables rigorous, spatially-varying LSFM deconvolution where it once was computationally infeasible. All that remains is to obtain the 3D PSF at each point along the light-sheet focusing axis. This can be done in multiple ways using only a single calibration volume; the experimental PSFs from this volume are either interpolated at the missing locations or are used to fit the parameters of a PSF model such as the Gibson-Lanni model~\cite{gibson1991experimental, li2017fast}. We tried both and found that the interpolation method is better for sparser samples, while the fitted model is better for dense samples. Details along with the derivation of sheet convolution/deconvolution are in Sec.~\ref{methods:sheet} and simulations are in Appendix~\ref{appendix:sdm}.

To combat the inherent spatial-variance of LSFM, practitioners typically employ tiling LSFM~\cite{gao2015extend}, in which only the thin lateral slab of the image that lies within the light-sheet waist is kept. To acquire the full FoV, the sample is shifted after each acquisition to ensure that each part of it is imaged within the light-sheet waist. The final image is then stitched together from multiple acquisitions. However, by using our sheet deconvolution approach, the parts of each acquisition that lie outside of the waist can be recovered computationally, thereby expanding the usable FoV of the LSFM system and reducing the number of acquisitions needed for an object with a large lateral extent. While there are techniques, like axially-swept LSFM~\cite{dean2022isotropic}, which speed up the process in hardware, a purely computational solution is desirable for its ease of use, reduced irradiation, and flexibility. 

Figure~\ref{fig:light-sheet} shows a demonstration of sheet deconvolution on LSFM data. For the experiment, a 3D stack of randomly-scattered fluorescent nanospheres in agarose gel was imaged and used to obtain PSFs at each position along the spatially-varying axis (see Sec.~\ref{methods:sheet}). Then, using these PSFs, sheet deconvolution is applied to samples of two types: fluorescent beads and SU8686 cells. Unlike standard deconvolution, sheet deconvolution is able to better resolve features that do not lie in the light-sheet waist and create a more homogeneous axial resolution across the entire FoV than standard deconvolution does. We quantified the resolution improvement by measuring the lateral and axial full-width half-maximum (FWHM) of fluorescence beads across the entire FoV after applying the two deconvolution methods. As expected, the lateral FWHM does not change significantly since the PSF spatial variance occurs primarily in $z$, but the axial FWHM values from sheet deconvolution are on average $~300$ nm smaller than those from standard deconvolution with $~200$ nm less standard deviation. This increase in resolution outside of the beam waist can reveal subcellular-scale features that were previously only accessible by shifting the sample. In this experiment, sheet deconvolution ran on a volume size of $512 \times 512 \times 160$ in about 7 minutes, which is orders of magnitude faster than solving a spatially-varying deblurring method without leveraging symmetry. With sheet deconvolution, users of LSFM can speed up the capture of large samples while still getting high resolution across the entire FoV. 

\subsection{Extensions and Variations}
\label{sec:results:extensions}

We have demonstrated our methods for rotational and lateral symmetry, but expect that the extension to other symmetries would be analogous. In addition, elements of RDM may find use even when the system is not space-varying, or when calibration data is not available.

\subsubsection*{Space-invariant systems}
If aberrations are not space-variant, standard deconvolution should perform as well as RDM and be more computationally efficient. However, the first part of the RDM pipeline can still provide value by estimating the spherical aberration coefficient from a calibration image via Seidel fitting. With this coefficient, we can generate a synthetic center PSF and perform deconvolution. We call this procedure \emph{Seidel deconvolution}, and find that it essentially denoises the PSF measurement since it finds the closest synthetic PSF to the measured one. In the results shown in Fig.~\ref{fig:4}, Seidel deconvolution resolves smaller features and gives a cleaner reconstructed image than standard deconvolution with the measured PSF. Fitting synthetic PSFs for the purpose of deconvolution is not new; existing work has fit a variety of parametric models to the experimentally measured PSF including a 2D Gaussian distribution~\cite{YangGaussian}, Gaussian mixture model~\cite{samu}, Zernike basis~\cite{Zheng:13, Shao:19, Maalouf:11,Bech, hanser2003phase, hanser04}, spherical aberration diffraction model~\cite{Panka}, and Seidel ray model~\cite{Wang_radial_ray, Simpkin}, but to our knowledge, the fitting of Seidel polynomials in the pupil function is novel. 

\begin{figure}[tbh!]
 \centering
\includegraphics[width=\linewidth]{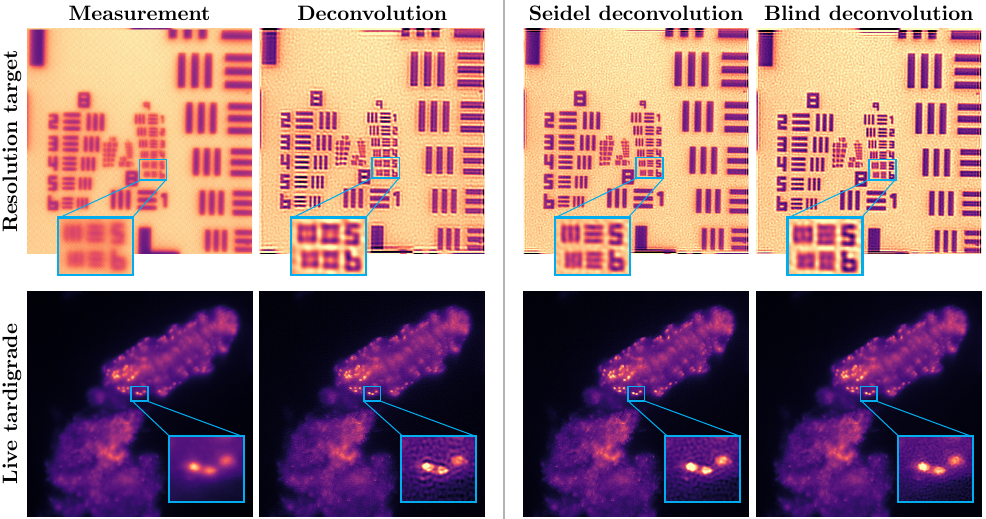}
 \caption{\textbf{Variations of RDM for space-invariant systems.} We compare results from our Seidel and blind deconvolution algorithms to that of standard deconvolution for images with a USAF test target and live tardigrades. Even for space-invariant systems, our methods outperform standard deconvolution by using a synthetic PSF, which prevents artifacts and loss of resolution from noise-based artifacts. Our blind deconvolution method usually estimates the spherical Seidel coefficient well; however, for certain images, the blind method can overestimate the coefficient, leading to oversharpening of the image, as in the USAF resolution chart.
 }
 \label{fig:4}
\end{figure} 

\subsubsection*{Blind deblurring}
One intriguing extension of RDM is to perform deblurring without calibration data (\emph{blind deblurring}). Such a situation may arise because measuring PSFs is difficult, inconvenient, or not available for images captured in the past. Blind deblurring involves joint estimation of the PSF and the deblurred object from a single image (see Methods~\ref{methods:blind}). In principal, this can be done for any spatially-varying system; however, it is computationally intensive, so we demonstrate here only the simple case of a space-invariant system. Thus, only the spherical Seidel coefficient needs to be estimated. Results are shown in Fig.~\ref{fig:4} along with non-blind deconvolution methods for comparison. The blind deconvolution result is similar to that of Seidel deconvolution, though cannot resolve the smallest features and has overly high contrast due to noise. Future work extending this idea to spatially-varying systems could use DeepRD, which is more computationally efficient than ring deconvolution, to iteratively search over the space of deblurring networks and choose the network with the sharpest reconstruction.



\section{Discussion}
In summary, we developed a new pipeline for image deblurring in rotationally-symmetric systems, called  Ring Deconvolution Microscopy (RDM). RDM encompasses both an accurate and analytically-derived deblurring technique, and its fast alternative, Deep Ring Deconvolution (DeepRD). Like standard deconvolution microscopy, our methods only require a single calibration image; however, they offer space-varying aberration correction. We back RDM with a new theory of imaging under rotational symmetry, which we call linear revolution-invariance (LRI), and an implementation of the LRI forward model (ring convolution). For RDM calibration we also develop a procedure for fitting Seidel aberration coefficients from a single calibration image of randomly-placed point sources. In order to show generality of our ideas, we further derive, implement, and test an analogous method that exploits linear (instead of radial) symmetry, for applications in light-sheet microscopy. We verify the accuracy of our deblurring methods in both simulation and experimentally over four diverse microscopy modalities. 

We hope that RDM will ultimately replace deconvolution microscopy as standard practice in widespread applications from biology to astronomy. We believe that RDM will find most use in systems that approach optical extremes such as miniature microscopes or large FoV systems, but may also empower optical designers to simplify hardware knowing they have the ability to better correct for aberrations digitally. RDM is well-suited to dynamic conditions in which the system or sample is changing, as long as the system calibration can be updated accordingly. For simpler cases that can be directly modeled, such as the time-varying deformation of an imaging fiber, it should be possible to update the initial calibrated PSFs with a theoretical model (~\cite{ploschner2015seeing, wen2023single}). 

We intend RDM to be a living, breathing tool with constant improvements to its speed and accuracy. Already we have seen a twenty-fold decrease in ring deconvolution's runtime in preliminary experiments by parallelizing it over multiple GPUs using the novel Chromatix optical framework~\cite{chromatix_2023}. Moreover, the constant improvement in deep learning architecture, including better conditional models~\cite{chauhan2024brief}, can also improve the performance of DeepRD in future. We plan to continually update our \href{https://github.com/apsk14/lri}{codebase} with these improvements.

\section{Methods}
\label{sec:methods}
\subsection{Ring Convolution and Ring Deconvolution}
    \label{subsec:ring}
    We begin with a primer on notation.
    Let $g(u,v)$ describe the object's intensity at $(u,v)$ and $h(x,y;u,v)$ describe the space-varying PSF - the intensity at $(x,y)$ of the PSF generated by a point source at $(u,v)$.
    We further use the notation $\tilde{g}$ to denote the transformation of $g$ to polar coordinates.
    Then the final image intensity  $f(x,y)$ of a linear optical system is formed by the superposition integral~\cite{goodman2005introduction}: 
	\begin{equation}
		\label{eq:lin}
		f(x,y) = \int \int g(u,v)h(x,y;u,v)dudv.
	\end{equation}
    This equation is the system \emph{forward model} for a linear space-varying system---it describes the image as a function of the object and PSFs at differnt locations in the FoV.
    
    standard image deconvolution approximates the system as linear space-invariant (LSI), which means the PSF is the same at all positions in the FoV, $h(x,y;u,v) = h(x-u,y-v;0,0)$. This simplifies the forward model in Eq.~\eqref{eq:lin} by reducing it to a convolution with a single PSF. This greatly simplifies the computation for forward and inverse problems, but at the cost of being inaccurate for space-varying aberrations.    

    In this paper, we incorporate radially-varying aberrations analytically into the forward model in order to provide a middle ground between purely space-invariant and completely space-varying systems. The assumption is that the system is \emph{linear revolution invariant} (LRI) - i.e. its physical configuration is symmetric about the optical axis (see Fig.~\ref{fig:1}a). This is true of many typical optical imaging systems.
    Our core observation is that all LRI optical systems satisfy
    \begin{equation}
        \label{eq:lri-assn}
        \tilh(\rho, \phi; r, \theta) =  \tilh(\rho, \phi - \theta; r, 0).
    \end{equation}
        
    Under this assumption, the object intensity can be written as what we call a \emph{ring convolution}, denoted by $ f \triangleq g \lri h$. By substituting~\eqref{eq:lri-assn} into \eqref{eq:lin} with $x = \rho cos \phi$ and $y = \rho sin \phi$, we get
    \begin{equation}
        \label{eq:ring-convolution}
        f(x,y) = (g \lri h) (x,y) =  \int_0^\infty r (\tilde{g} *_{\theta} \tilh ) (\sqrt{x^2+y^2},\tan^{-1}(y/x);r,0) dr,
    \end{equation}
    where the $*_{\theta}$ operator indicates a one-dimensional convolution over the $\theta$ dimension. The $r$ arises in the deconvolution since we are integrating over object space  $(u,v)$ in polar coordinates $(r,\theta)$. This ring-wise computation, wherein points at different radii are filtered heterogeneously, is consistent with the underlying intuition in LRI: the blur varies radially.
    Our first main result is an efficient, FFT-based inversion of the ring convolution.
    \begin{theorem}[Ring Convolution Theorem]
    	\label{thm:lri_conv}
    	Under LRI, where $\F_\Theta$ is a 1D Fourier transform over $\theta$, 
    	\begin{equation}
    		\label{eq:lri_conv}
    		\tilf(\rho,\phi) = \F_\Theta^{-1}\Big\{\int r\F_\Theta\{\tilde{g}(r,\theta)\}\F_\Theta\{\tilh(\rho, \theta; r)\}dr\Big\}(\phi).
    	\end{equation}
    \end{theorem}
    \begin{proof}
        \renewcommand{\qedsymbol}{}
        Since the given system is LRI,~\eqref{eq:ring-convolution} holds. 
        Substituting $\rho = \sqrt{x^2 + y^2}$ and $\phi = \tan^{-1}(y/x)$, we can rewrite~\eqref{eq:ring-convolution} as
        \begin{equation}
            \tilf(\rho,\phi) =  \int_0^\infty r (\tilde{g} *_{\theta} \tilh ) (\rho,\phi;r,0) dr.
            \end{equation}
        Applying the Fourier Convolution Theorem to the one-dimensional convolution on the right-hand side yields
        	\begin{equation}
        		\tilf(\rho,\phi) = \int r\F_\Theta^{-1}\Big\{\F_\Theta\{\tilde{g}(r,\theta)\}\F_\Theta\{\tilh(\rho, \theta; r)\}\Big\}(\phi)dr,
        	\end{equation}
        
        \noindent where $\F_\Theta$ is the one dimensional Fourier transform over $\theta$. 
        By Fubini's theorem, we pull the inverse Fourier Transform outside of the integral, which gives
        \begin{equation}
        		\tilf(\rho,\phi) = \F_\Theta^{-1}\Big\{\int r\F_\Theta\{\tilde{g}(r,\theta)\}\F_\Theta\{\tilh(\rho, \theta; r)\}dr\Big\}(\phi). \; \; \; \; \blacksquare
	\end{equation}

    \end{proof}
    Ring convolution has an efficient and convex formulation for computing its inverse, ring deconvolution:
    \begin{equation}
    \label{eq:lri_opt}
    	\hat{g} = \argmin_{\bar{g}}|| f - \bar{g}  \lri h||_2^2.
    \end{equation}
    This problem can be solved efficiently via an iterative least squares solver using Algorithm~\ref{alg:lri_forward} as a substep. A Fourier interpretation of ring convolution is provided in Appendix~\ref{appendix:fourier}. While the above results are all rigorous, the discrete time implementations of them have small, but nonzero errors due to discretization. For example, the polar transformation in Algorithm~\ref{alg:lri_forward} requires a small amount of interpolation. As is the case with  standard deconvolution, there are conditions for which ring convolution is not fully invertible and consequently ring deconvolution will not recover the sample exactly. The diffraction-limit, for example, manifests in PSFs whose frequency spectrum is bandlimited, rendering  frequencies in the sample that are past the bandlimit irrecoverable. This can also happen if   certain frequencies in the system's transfer function are below the noise floor. In such cases ring deconvolution, because it is convex, will return an estimate of the sample that is closest in $l2$ norm to the true sample. Regularization can improve this result even further by leveraging prior knowledge about the sample in question.  
    
    \begin{algorithm}[t]
\caption{Ring convolution} 
	\label{alg:lri_forward}

	\textbf{Input:} $N\times N$ pixel image $g$; PSFs along one radial line $h^{(j)}$, $j = 1, \dots, K$; corresponding distances $r_j$, $j = 1, \dots, K$ of each PSF from the center.
	
	\textbf{Output:} LRI blurred image $f$
	
	\begin{algorithmic}[1]
		\State $\tilde{g} \gets$ polarTransform($g$) \Comment{polar dimensions are $M \times K$, angle by radius} 
		\State $\tilde{f} \gets$ zeros($M\times K$) \Comment{initialize the output in polar form as an all zero matrix}
		
		\For{$j = 1, \dots, K$}
		\State $\tilde{h}^{(j)} \gets$ polarTransform($h^{(j)}$) 
		
		\For {$i = 1, \dots, K$}
		 \State $\tilde{f}_{:,i} \gets \tilde{f}_{:,i} +  \ifft\{ r_j \fft \{\tilde{g}_{:,j}\} \fft \{\tilde{h}^{(j)}_{:,i}\}\}  $ \Comment{compute polar output ring by ring, FFT is 1D}
		
		\EndFor
		\EndFor
		
		\State $f \gets $ inversePolarTransform($\tilde{f}$)
	\end{algorithmic}
	
    \end{algorithm}

    \subsection{Fitting Seidel Coefficients to PSFs}
    \label{subsec:seidel}

    Ring (de)convolution algorithms require $h$, the collection of PSFs along one radial line in the FoV.
    Fortunately, there is a convenient and compact alternative to measuring these manually.
    The Seidel coefficients~\cite{born2013principles,voelz2011computational} are a polynomial basis that can represent any rotationally-symmetric system.
    We mathematically describe the form of these aberrations in Appendix~\ref{appendix:seidel}.


    The estimation procedure for estimating the Seidel coefficients involves fitting them to a single, sparse image of a few randomly scattered point sources (e.g. fluorophores on a microscope slide).
    Such an image is usually easier to obtain than an image of an isolated point source in the center of the FoV. The presence of off-axis PSFs in the calibration image provides information about all aberration coefficients. Though it may be possible to fit these coefficients from a single off-axis PSF, we find that a few, randomly placed PSFs provides a robust fit. The locations of the point sources are not known \textit{a priori} and are estimated via local peak detection. The optical center of the system is also needed in order to properly localize the PSFs; this can be found heuristically or by using common optical center finding algorithms such as~\cite{5206746}.
    
    Let $r_1, \ldots, r_J$ be the object-plane radii of the $J$ points in the calibration image. 
    We then find the primary Seidel coefficients $\hat{\omega}$ whose generated PSFs best match the measured PSFs. 
    Once again, this searching procedure is succinctly stated as an optimization problem,
    \begin{equation}
        \label{eq:seidel-fit}
    	\hat{\omega} = \argmin_{\bar{\omega}} \sum_{j=1}^J ||h^{(j)} -  \F^{-1}\{ P({\bar{\omega}})^{(j)} \}||_2^2,
    \end{equation}
    \noindent where $P({\bar{\omega}})^{(j)}$ is the pupil function with Seidel coefficients $\bar{w}$ from a point source at distance $r_j$ from the center of the FoV.
    
    It has been shown that for LRI systems, the optimal fit $\hat{w}$ achieves a diminishing error~\cite{born2013principles}.
    Furthermore, the 5 primary Seidel coefficients index the dominant aberrations present in practical imaging systems: sphere, coma, astigmatism, field curvature, and distortion. For more complex aberrations it is possible to add higher-order Seidel coefficients to the fit. See Appendix~\ref{appendix:higher-order} for an exploration of the higher-order terms.
    In practice, we fit these 5 coefficients via the ADAM optimizer~\cite{kingma2014adam} and obtain reasonable local minima even though the problem is nonconvex (see Fig.~\ref{fig:2}).
    Armed with the estimated Seidel coefficients, we can generate synthetic PSFs at any radius. 
    
\subsection{Deep Ring Deconvolution (DeepRD)}
\label{subsec:deeprd}
DeepRD is designed to incorporate the Seidel coefficients into the deblurring process in a parameter-efficient and interpretable manner. To that end, we propose a neural network architecture inspired by the physical LRI image formation model. The first key design element is to use a modified Hypernetwork~\cite{ha2017hypernetworks}, a network which predicts the weights of another, task-specific ``primary" network. In DeepRD, an MLP-based hypernetwork takes in Seidel coefficients and produces a deep deblurring network that specifically works for the given coefficients. Our second key design element is the use of ringwise convolution kernels. Specifically, the hypernetwork produces CNN kernels for each radius which the primary network applies ring-by-ring. This replicates the revolution-invariance assumption central to ring deconvolution and eases the space-invariant constraint of typical convolutional kernels. Together, this design enables a neural network that is a fraction the size of a conventional U-Net with improved performance and generalization. To our knoweledge, we are the first to utilize learnable ring convolutions in the context of deep learning.

To produce a training dataset for DeepRD, we synthetically generate blurred input images from the Div2k and CARE fluorescence microscopy datasets using ring convolution with randomly sampled sets of Seidel coefficients. We must sample coefficients which adequately cover the attributes of realistic imaging systems. Each coefficient, which is in units of waves, is drawn from a uniform distribution and a noise-perturbed grid in the range $0-3$ waves---we emperically find this range to cover the aberrations of systems ranging from perfect to highly aberrated. Note that without ring convolution, such a dataset would be prohibitively slow to produce. With that in mind, we will release both an implementation of ring convolution as well as our dataset.

DeepRD accepts a two-part input, a blurred image and its corresponding primary Seidel coefficients. The model is then supervised with the unblurred ground truth image. We find that this physically-grounded synthetic dataset generation is effective to train models that generalize to real-world evaluation.

\subsection{Sheet Convolution, Deconvolution, and the LSFM PSF Model}
    \label{methods:sheet}

    Here, we derive the efficient sheet convolution and deconvolution mathematically.
    The derivations are roughly similar to those for ring convolution, but with a different symmetry.
    Because of the similarities, our exposition here is more terse.

    \subsubsection{Sheet convolution and deconvolution}
    LSFM is a three-dimensional imaging modality and therefore the object and PSF will be a function of three space variables. Let $g(u,v,t)$ describe the object intensity at location $(u,v,t)$, and $h(x,y,z;u,v,t)$ describe the space-varying PSF due to Gaussian light-sheet illumination focused along the $u$ dimension (the choice of $u$ as the first spatial dimension is arbitrary). The superposition integral for a linear optical system can then be written as  
	\begin{equation}
		\label{eq:sheet-lin}
		f(x,y,z) = \int \int \int g(u,v,t)h(x,y,z;u,v,t)dudvdt,
	\end{equation}
    analogously to~\eqref{eq:lin}.
    As before, we will incorporate the symmetry of LSFM to simplify the above display.
    The symmetry assumption states that the imaging optics are space-invariant, but the PSF varies in the $u$ direction due to the beam profile.
    Recall also that the total PSF is the product of the imaging PSF with the illumination PSF, which is varying.
    The following equation encodes these assumptions:
    \begin{equation}
		\label{eq:sheet-psf}
		h(x,y,z;u,v,t) = h(x,y-v,z-t;u,0,0).
	\end{equation}
    As before, plugging this into the linear forward model gives us sheet convolution:
    \begin{equation}
		\label{eq:sheet-conv}
		f(x,y,z) = \int \int \int g(u,v,t)h(x,y-v,z-t;u,0,0)dudvdt = \int (g *_{v,t} h)(x,y,z;u,0,0) du,
	\end{equation}
    where $*_{v,t}$ represents a 2D convolution along the $v$ and $t$ axes.
    From the equation we see that the image is the integral of two-dimensional convolutions of $y-z$ sheets of the object with $y-z$ sheets of the PSFs.
    To compute this integral, we only need to know the PSF at all values of $u$, and the $v$ and $t$ dimensions do not matter.

    With a forward model in hand, sheet deconvolution can be solved exactly the same way as ring deconvolution: iterative least squares (Eq.~\eqref{eq:lri_opt}). The discretization and computational implementation of these algorithms mimic those of ring convolution and can be found in our \href{https://github.com/apsk14/lri}{codebase}.

    \subsubsection{LSFM PSF model}
    The remaining component of LSFM deblurring strategy, is to obtain the necessary set of three-dimensional PSFs along one dimension. Similar to our earlier strategy for calibrating ring convolution/deconvolution, we can do so by imaging a single calibration volume containing randomly located, sparsely distributed point sources (e.g., fluorescent beads embedded in agarose). The resulting image stack will contain a random set of PSFs at different $u$ locations. 
    
    The simplest option is to directly estimate the PSFs at the missing $u$ locations by taking the convex combination of the two closest measured PSFs from the calibration stack. We use this strategy for the bead experiment in Fig.~\ref{fig:light-sheet}.

    The second option is more akin to our Seidel fitting procedure; it involves parametrization of the spatially-varying PSFs with a unified, differentiable model. In this case we develop a modified version of popular Gibson-Lanni three-dimensional PSF model for LSFM~\cite{gibson1991experimental}.  We will focus on our modifications of the model; further details about the Gibson-Lanni model and its variants are ubiquitous in the literature~\cite{gibson1991experimental, li2017fast}. Given a vector $p$ of system characteristics (e.g., sample refractive index) the Gibson-Lanni model calculates the optical path difference (OPD) between the ideal and experimental imaging systems. Integrating over this OPD gives the system's three-dimensional PSF. Our LSFM PSF model takes this PSF and truncates its $z$ extent according to the light-sheet illumination thickness, thereby creating a spatially-varying PSF in $u$. Formally, let the Gibson-Lanni PSF be $h_{p}$, then our LSFM PSF
    \begin{equation}
        h(x,y,z;u,v,t) = h_{p}(x-u,y-v,z-t)\frac{1}{\sigma(u) \sqrt{2\pi} } e^{-\frac{1}{2}\left(\frac{z-t}{\sigma(u)}\right)^2} ,
    \end{equation}
    where 
    \begin{equation}
        \label{eq:spread}
        \sigma(u) = \alpha \sqrt{1 + (\beta u)^2 }.
    \end{equation}
    The above equations arise from modeling the light-sheet as a Gaussian beam along $u$; its spread $\sigma(u)$ changes hyperbolically along the focus direction and its profile at a given $u$ is a Gaussian with variance $\sigma(u)^2$. We have two control knobs for it, $\alpha$, which controls the $z$ spread of the PSF at the thinnest section or waist, and $\beta$, which controls the rate that the spread increases as a function $u$. Given the calibration stack, we can optimize $\alpha$ and $\beta$ such that our PSF model produces PSFs close to the experimental ones. Since the PSF is a differentiable function of $(\alpha, \beta)$, they can be solved for using gradient-based iterative optimization, just like for Eq.~\eqref{eq:seidel-fit}. It is also possible to incorporate elements of $p$ into this optimization if they are not known \textit{a priori}. This model is used to calibrate sheet deconvolution on the cell sample in Fig.~\ref{fig:light-sheet}.

    In order to deploy the above models experimentally, one must ensure that the point sources are sufficiently sparse such that the PSFs are mostly non-overlapping. The exact point source locations are not important and can be estimated. Noisy PSFs are best handled with the Gibson-Lanni fitting method, which uses synthetic PSFs but still fits the measured PSFs well (see Appendix~\ref{appendix:sdm}). The interpolation method is more sensitive to noise but can be denoised effectively with thresholding and median filtering. The fitting procedures are detailed in our \href{https://github.com/apsk14/lri}{codebase}, and follow the same general structure: first take a calibration stack of randomly-scattered point sources, estimate the locations of the PSFs using local peak finding algorithms, produce generated PSFs at those locations, and use the error between the generated and measured PSF to update the PSF model. In the case of interpolation, the last step is replaced by linearly interpolating the two closest measured PSFs to the desired location to create the generated PSF there.

\subsection{Blind Deblurring}
\label{methods:blind}
Our version of blind deconvolution also takes advantage of the Seidel coefficients. Given just a blurry image, we start by randomly picking a value for the spherical Seidel coefficient. Then we use this value to synthetically generate a point spread function, and use it to deconvolve the blurry image. We then compute the sum of the spatial gradient of the resulting deconvolved image (this acts as a surrogate measure of image sharpness) and use its negative as a loss. We then minimize this loss (maximize the sharpness) by updating our initial guess of the spherical aberration coefficient using its gradient with respect to the loss function. Running this iteratively, we eventually converge to a final spherical aberration coefficient, generate a final synthetic PSF, and deconvolve the blurry image with this PSF to get the final result.

Note that this procedure generalizes to spatially-varying systems. We would instead jointly estimate all 5 Seidel coefficients and use ring deconvolution instead deconvolution at each step. This, however, is computationally expensive and requires generating $N$ (the image sidelength) PSFs per iterative step. We believe it is possible to do this more efficiently with DeepRD by replacing the ring deconvolution operation at each step with DeepRD. That is, we search the space of DeepRD networks for the one that produces the sharpest reconstruction. However, this is out of scope for this project, and we leave it as future work.

\subsection{Experimental Details}
Experimental details for the micro-endoscopy experiment can be found in~\cite{Turcotte:20}.

\label{subsec:exp_deets}
\subsubsection*{Sample Preparation} 
\noindent{\textbf{Live tardigrades}}
Tardigrades were mixed-staged adults of the eutardigrade species Hypsibius exemplaris Z151 (reclassified from Hypsibius dujardini in 2017), purchased from Sciento (Manchester, United Kingdom). Animals were cultured as described in~\cite{lyons2024survival}.
A mixture of starved and nonstarved tardigrades were stained overnight with Invitrogen nucleic acid gel fluorescent stain, whose excitation and emission maxima are 502 nm and 530 nm, respectively. Individual stained tardigrades were then isolated onto a glass slide for imaging. Meanwhile, the non-fluorescent samples (USAF resolution targets, and rabbit liver tissue) were obtained imaging-ready on glass slides. 

\noindent{\textbf{BPAE cells}}
Bovine pulmonary artery endothelial cells were obtained from Thermo Fisher. They are labeled with with MitoTracker™ Red CMXRos and Alexa Fluor™ 488 Phalloidin.

\noindent{\textbf{Beads}}
100nm Fluorescent beads were obtained from Polysciences (17150-10). For measuring the PSF in the light-sheet fluorescence microscope, we embedded the beads in 2\% agarose with a final density $~5x10^{-4}$ of the stock solution.

\noindent{\textbf{SU8686 cells}}
SU8686 cells labeled with F-tractin-mRuby were embedded in soft bovine collagen, and then fixed before imaging with light-sheet fluorescence microscope.

\subsubsection*{Imaging}
\noindent{\textbf{UCLA Miniscope}}.
We used the UCLA Miniscope v3 with the Ximea MU9PM-MBRD 12 bit, 2.2 micron pixel sensor. Optically, the Miniscope is comprised of a gradient index objective and achromat tube lens; further details are provided in \cite{aharoni}. In order obtain the system PSFs, we imaged $1 \mu$m fluorescent beads smeared on a glass slide. For deconvolution microscopy calibration, we repeatedly diluted the bead solution with isopropyl alcohol until we were able to sufficiently isolate a single bead, whereas for RDM calibration we used a single dilution and imaged a slide containing a sparse collection of beads. We used a custom 3D Prior stage controlled with Micromanager v1.4 and Pycromanager~\cite{pinkard2021pycro}.

\noindent{\textbf{Multicolor fluorescence microscope}}. 
We used a Nikon Plan Apo VC 100x Oil DIC N2 objective with 1.518 refractive index oil in a Nikon Eclipse Ti2 controlled with the Nikon NIS Elements Software. Images were taken with a Hamamatsu Orca Flash 4.0 camera with $0.065 \mu m$ pixel pitch. The PSFs were obtained with $0.01 \mu m$ FluoSpheres Yellow-Green$ 505/515 nm$ F8803 and FluoSpheres Red $580/605 nm$ F8801 beads. First, we diluted beads in water and then further in ethanol until sufficient sparsity was achieved. The bead solution was then smeared on a slide and left to dry. Finally, the beads were mounted with a drop of glycerol and sealed with nail polish.

\noindent{\textbf{Light-sheet fluorescence microscope}}.
We used a previously published setup for Field Synthesis~\cite{chang2019universal} that was operated without a ring mask, rendering it to a multi-directional selective plane illumination mSPIM~\cite{huisken2007even} system with a Gaussian sheet. The microscope equipped with 488 and 561 nm laser illumination, a Special Optics 28.5X/NA 0.67 illumination objective, and a Nikon 25X/NA 1.1 detection objective, and is controlled with a custom LabView program written by Coleman Technologies and is equipped with temperature control for long term live cell imaging.

\subsubsection*{Image processing}
All experimental images were captured and stored in a raw, unprocessed format (npy or tif). Miniscope images underwent hot pixel removal (detailed in the public code) and normalization prior to deblurring. These images were cropped afterward by 10 pixels in each dimension to remove edge artifacts. Multicolor images were downsampled by a factor of 2, separated into two channels and deblurred independently. After deblurring, the channels were recombined and globally contrast stretched for display. Pseudocoloring was done with ImageJ using the Green/Magenta LUT. These images were also cropped for edge artifacts. The details of the multimode fiber images can be found in~\cite{Turcotte:20}. The bead images were upsampled by 3x and convolved with a Gaussian kernel (3/2 pixel width) after deblurring. Neuron images were convolved with a Gaussian kernel (1 pixel width) after deblurring. This was done according to~\cite{Turcotte:20}. LSFM stacks were similarly cropped and contrast-stretched equally for each method for the purpose of display. For simulation data, images were normalized before deblurring and cropped after deblurring. All displayed PSFs were globally contrast stretched for display.

\subsubsection*{Computation}
PSF generation for the simulation experiments was done by synthetically generating pupil functions with the given Seidel coefficients \cite{voelz2011computational} (see~\href{https://github.com/apsk14/lri}{Github}). Computation was done using Python on a single GPU, either a  NVIDIA GeForce RTX 3090 or NVIDIA RTX A6000. For standard deconvolution the measured PSF was denoised through background subtraction and pixel-wise thresholding. For each $1024 \times 1024$ image from the Miniscope and high NA multicolor systems, ring deconvolution took about 115 seconds and DeepRD took about $125$ milliseconds. For each $512 \times 512$ image in simulation ring deconvolution took about 60 seconds and DeepRD took about $0.1$ seconds. For the $360 \times 360$ images from the micro-endoscope, ring deconvolution took about 20 seconds. For a single $512 \times 512 \times 160$ volume from the LSFM system, sheet deconvolution took about 7 minutes. 

All non-learning, iterative methods are solving linear least squares optimization problems (see~\eqref{eq:lri_opt}); we additionally add TV regularization to these and run them till convergence using an ADAM optimizer~\cite{kingma2014adam}. For each method, the hyperparameters---including learning rate and regularization strength---that provided the smallest loss and best qualitative results were used. For deconvolution we tried a variety of algorithms in addition to the iterative scheme, including Wiener filtering and Richardson-Lucy deconvolution, and used the best reconstructions, which was either iterative deconvolution or unsupervised wiener filtering~\cite{Orieux:10}.

Open source implementations of ring convolution, polar transform, Seidel fitting, and ring deconvolution as well as the light-sheet extension methods can be found in the \href{https://github.com/apsk14/lri}{codebase}. Our intent is for this codebase to function as an easy-to-use library such that any practitioner with any imaging system can utilize RDM with little-to-no overhead.
    
The baseline U-Net and DeepRD were both trained on ring-convolved images from the Div2k dataset. For the simulation results, both models were additionally fine-tuned on images from the CARE dataset. All models were trained till convergence of the validation loss and optimized over hyperparameters. 

\section{Code Availability}
The code for implementing ring convolution, ring deconvolution, DeepRD (including pretrained model weights) and Seidel PSF fitting along with tutorials on our experimental data are publicly available on Github (\url{https://github.com/apsk14/rdmpy}).

\section{Data Availability}
The data used in all of the imaging experiments (Miniscope, multicolor fluorescence, multimode fiber, and light-sheet) is publicly available on Box (\url{https://berkeley.box.com/s/zmsjjgmquwq2roh4d9qthcnv3rhwuidn}). Additional experimental data from the multimode fiber system can be requested from~\cite{Turcotte:20}(\url{https://opg.optica.org/boe/fulltext.cfm?uri=boe-11-8-4759&id=433935}). The datasets used to train and fine-tune DeepRD, and to evaluate the quantitative performance of the methods are also hosted on Box (\url{https://berkeley.box.com/s/vv3g6avhrr9agijmlj3b1153oo7x9gao}). These datasets were sourced from the CARE dataset~\cite{weigert2018content}  (\url{https://publications.mpi-cbg.de/publications-sites/7207/}) and the Div2k dataset~\cite{agustsson2017ntire}  (\url{https://data.vision.ee.ethz.ch/cvl/DIV2K/}). The high resolution pretraining dataset, due to its large memory useage, will be made available upon request. 

\section{Inclusion and Ethics}
This study included contributions from local researchers at UC Berkeley, UC San Francisco, UT Southwestern, and Nikon Imaging Center at Harvard Medical School. Data was acquired by local researchers at each location. These researchers are either included as authors or acknowledged below.

\section*{Acknowledgements}
We would like to thank the following organizations and individuals: Ana Lyons and Saul Kato's lab for providing the tardigrades for imaging, the authors of~\cite{Turcotte:20} for giving us access to the micro-endoscopy data, the Nikon Imaging Center at Harvard Medical School for their imaging support, Neerja Aggarwal at UC Berkeley for help and access to 3D printing, Gabriel Muhire Gihana from Danuser lab at UTSW for providing SU8686 cells, and Diptodip Deb from Janelia Research Campus in Virginia for running preliminary multi-GPU experiments.

A.K. was funded by the Berkeley Fellowship for Graduate Study. 
A.N.A. was supported by the Berkeley Fellowship for Graduate Study and the National Science Foundation Graduate Research Fellowship Program under Grant No. DGE 1752814. Any opinions, findings, and conclusions or recommendations expressed in this material are those of the author(s) and do not necessarily reflect the views of the National Science Foundation.
S.Y. is supported by Moore Foundation data-driven discovery.
This material is based upon work supported by the Air Force Office of Scientific Research under award number FA9550-22-1-0521. 
This publication has been made possible in part by CZI grant DAF2021-225666 and grant DOI https://doi.org/10.37921/192752jrgbn from the Chan Zuckerberg Initiative DAF, an advised fund of Silicon Valley Community Foundation (funder DOI 10.13039/100014989).

\printbibliography

\newpage

\appendix

\section{Additional comparisons}
\label{appendix:additional-comps}
Two other commonly used deconvolution methods are Gaussian kernel deconvolution and modal decomposition. Gaussian kernel deconvolution involves fitting a two-dimensional Gaussian distribution to an experimental PSF and using the fitted kernel for deconvolution. While common, this method assumes spatial-invariance and cannot correct spatially-varying aberrations. On the other hand, modal decomposition is a spatially-varying deconvolution technique which models the spatially-varying image formation as a weighted sum of convolutions with space-invariant kernels. These kernels and their corresponding weights can be estimated from a uniformly spaced set of PSFs via singular value decomposition. We note that this type of calibration generally cannot be done single-shot and requires multiple acquisitions with a high precision motion stage. Since the multimode fiber experiment has such calibration information, we use it as a test bed for these additional techniques. Figure~\ref{fig:grid-comp} shows these methods on a grid of beads taken with the same multimode fiber system used in~\cite{Turcotte:20}.

Ring deconvolution exhibits the best performance among these methods. Gaussian kernel deconvolution suffers from a lack of detail in its PSF compared to the experimental one. Meanwhile modal decomposition, despite using a grid of 21 by 21 PSFs, is unable to properly interpolate the PSFs in areas between the calibration points and provides erratic performance throughout the FoV. It is also worth noting that that modal decomposition could be combined with our Seidel fitting procedure to further improve its accessibility; this exploration however is beyond the scope of this work and is left for future exploration.

\begin{figure}[h]
 \centering
\includegraphics[width=\linewidth]{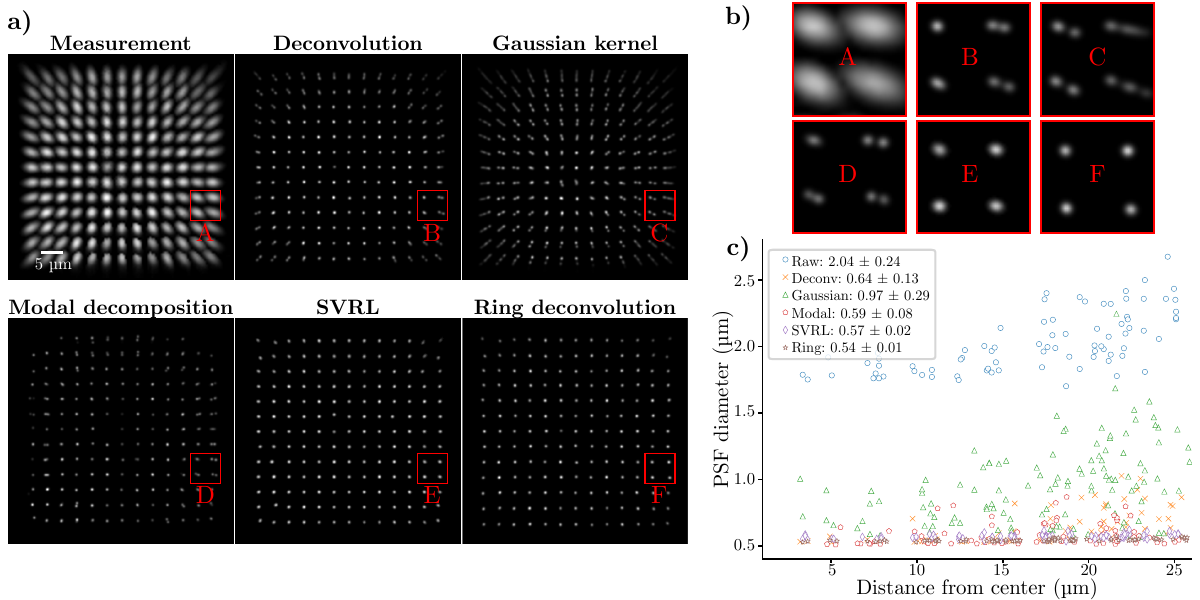}
 \caption{\textbf{Comparisons of various deblurring methods.} An evenly spaced grid of 13 x 13 point sources are imaged with the multimode fiber micro-endoscope from~\cite{Turcotte:20} and deblurred using a variety of methods. \textbf{a)} Results from each method. Top row (left to right) is the raw measurement, deconvolution with an experimental PSF and deconvolution with a Gaussian kernel fitted to the PSF. Bottom row (left to right) is deblurring via modal decomposition~\cite{yanny2020miniscope3d, flicker2005anisoplanatic}, SVRL~\cite{Turcotte:20}, and ring deconvolution. As seen in \textbf{b)}, the spatially-varying methods (bottom row) are superior to the deconvolution methods (top row) due to the substantial spatial variation in the PSF.
 Among the spatially-varying methods, ring deconvolution produces the smallest, most consistent beads. The quantitative results in \textbf{c)} also show that ring deconvolution has the smallest average bead radius with the least variation. Moreover, unlike the other methods, its performance does not degrade on beads far from the center of the FoV.
 }
 \label{fig:grid-comp}
\end{figure} 

\section{Quantifying the performance of SDM}

\label{appendix:sdm}

\begin{figure}[h]
 \centering
\includegraphics[width=\linewidth]{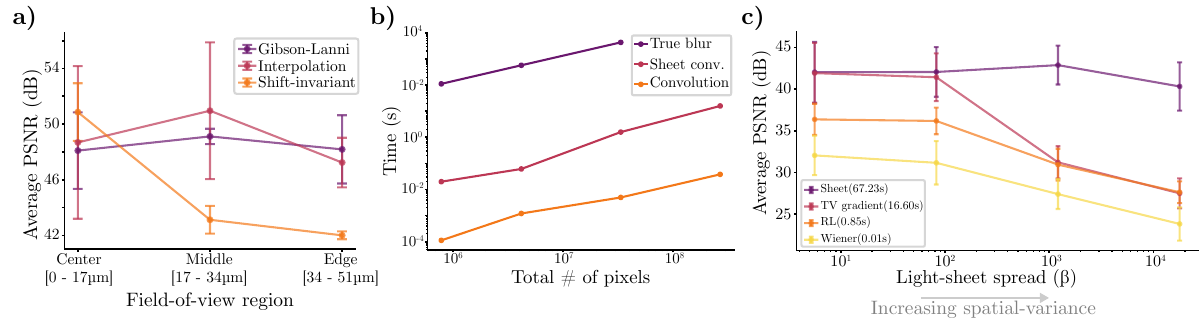}
\caption{\textbf{Quantifying Sheet deconvolution microscopy.} \textbf{a)} Error evaluation of PSF fitting models. N=6 image stacks of randomly-scattered point sources are acquired with a light-sheet microscope and fit using either with the modified Gibson-Lanni model, by interpolating the measured PSFs, or by using the (denoised) center PSF and assuming shift-invariance. The PSNRs between the generated PSFs and the measured PSFs from the remaining calibration image stacks is computed for all three PSF models. These PSNRs are stratified by regions split along the light-sheet spread direction; center is where the light-sheet is the thinnest, while edge is where the light-sheet is the most spread. The resulting average PSNRs and their standard deviations in each region are plotted as a function of region. \textbf{b)} Forward model computation time for different data sizes. The runtimes of three different light-sheet forward models are plotted as a function of data size; true blur (manually superimposing the shift-varying PSFs at every point in the 3D volumes), sheet convolution, and standard 3D convolution. Times are acquired by processing randomly-generated 3D volumes and are averaged over 10 trials. \textbf{c)} Quantitative evaluation of sheet deconvolution. N=20 3D volumes of randomly sized and oriented ellipsoids are generated. These objects are blurred via true blur using PSFs generated from our modified Gibson-Lanni model, and are corrupted with Poisson noise (SNR 15) to simulate light-sheet measurements. The simulated measurements are deblurred with sheet deconvolution, iterative gradient-based deconvolution with total variation regularization, Richardson-Lucy deconvolution, and Wiener deconvolution. The process is repeated with increasing values of the $\beta$ spread parameter (see~\ref{eq:spread}), beginning with a system that is shift-invariant ranging to a system which is highly shift-variant due to rapid light-sheet spread. The average PSNR over all 20 objects for each method is plotted as a function of the spread parameter.
 }
 \label{fig:sdm-quant}
\end{figure} 

The experimental evaluation of sheet deconvolution microscopy (SDM) (Sec.~\ref{subsec:ls}) highlights its resolution improvement and computational efficiency advantages. We further support these results by quantifying each part of the SDM pipeline. Figure~\ref{fig:sdm-quant} details these quantitative experiments. First, quantify how well our light-sheet PSF fitting models, modified Gibson-Lanni and interpolation, fit experimental PSFs and compare them to the standard shift-invariant model. To that end, we fit these models to experimental PSFs acquired from a light-sheet microscope and compute the PSNR between generated PSFs and measured PSFs. We plot these PSNRs as a function of the PSF’s location along the light-sheet spread, from the center, where the light-sheet is the thinnest, to the edges, where the light sheet is the most spread. Both of our models account for the shift-variance due to the light-sheet spread and consequently maintain accurate fitting across the entire FoV. In contrast, the standard shift-invariant model degrades away from the center, as expected. The interpolation model is higher variance than the modified Gibson-Lanni model; this is because the interpolation model is exact at the locations of the PSFs provided in the calibration image and worse in-between those PSFs, whereas the Gibson-Lanni fitting finds parameters that produce PSFs that fit well everywhere in the FoV, on average. 

Next, we measure the runtime of the corresponding forward models, sheet convolution and 3D convolution, on different data sizes and find that, while sheet convolution is slower than 3D convolution, it is much faster than computing the true blur, i.e., manually superimposing the PSF at each point in the 3D volume. The scaling of runtime with data size of sheet deconvolution and 3D deconvolution can be inferred from these results but ultimately also depends on the sample and noise properties. 

Finally, we quantify the accuracy of sheet deconvolution against standard deconvolution techniques, total variation (TV) iterative deconvolution, Richardson-Lucy deconvolution~\cite{Richardson, Lucy}, and Wiener deconvolution~\cite{Orieux:10}. To do this we simulate phantom 3D samples (randomly sized and oriented ellipsoids) and simulate a light-sheet measurement using the true blur model with Poisson noise. We then deblur the simulated measurement with each method and compute the PSNR of the result with the ground truth sample. We repeat this process for systems with increasing light-sheet spread (increasing spatial variance). When the system is shift-invariant, sheet deconvolution performs as well as the best deconvolution technique, TV regularized iterative deconvolution. As the system becomes increasingly spatially-variant, the performance of sheet deconvolution remains relatively constant, while the performance of all the deconvolution methods degrades severely.

\section{Simulation with Poisson noise}

\label{appendix:poisson}

The quantitative experiments from Fig.~\ref{fig:2} are rerun with Poisson noise instead of Gaussian noise. There are no notable differences in performances, with the same trends occurring as before: ring deconvolution is the best performing method, followed by DeepRD, the baseline U-Net, and finally standard deconvolution. Both deep learning models were updated with a small amount of training ($~100$ iterations) on Poisson-noised images from the training set. Ring deconvolution and standard deconvolution were run with additional total variation regularization to account for the Poisson noise. We found that an SNR range of 15-30 most accurately reflected experimental data and thus ran all simulations with the worst-case SNR of 15.

\begin{figure}[h]
 \centering
\includegraphics[width=\linewidth]{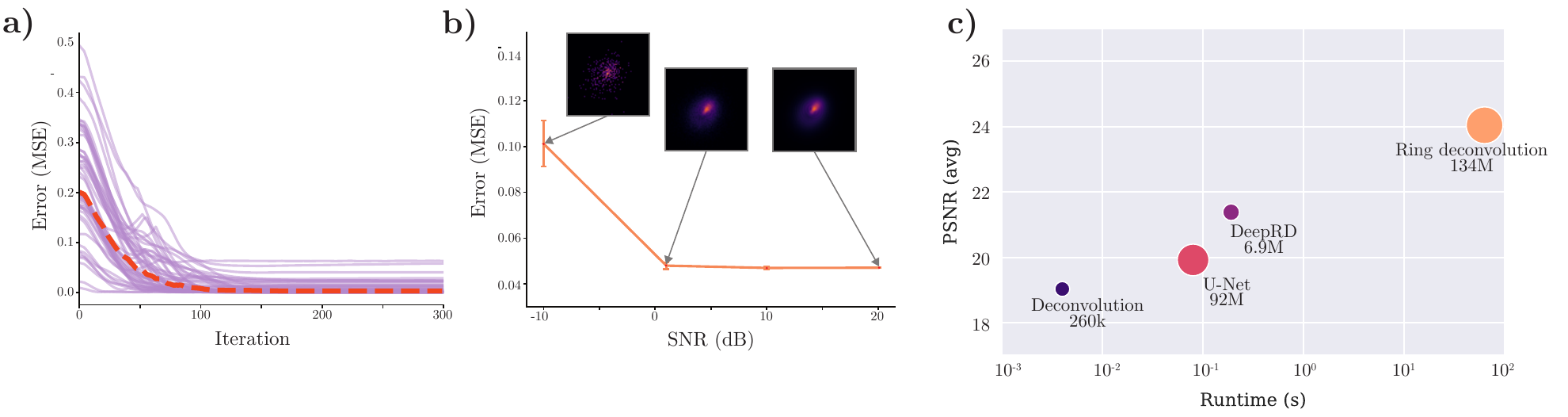}
 \caption{\textbf{Quantitative performance under Poisson noise.} \textbf{a)} A Poisson-noised image (SNR 15) of randomly placed PSFs is used to estimate the underlying Seidel coefficients, which can then be used to generate PSFs at any location. \textbf{e)} Seidel fitting error is plotted as a function of iteration in the optimization algorithm to demonstrate convergence. Each purple line is a different trial of n=50 trials with a randomly sampled set of underlying Seidel coefficients. The red dashed line is the per-iteration median.  \textbf{b)} Mean squared error of the fitted Seidel coefficients plotted as a function of SNR of the calibration image. The average coefficient error is plotted along with the variance (error bar) over n=50 random trials. Some example calibration PSFs are shown. Even under severe noise, the Seidel fit is still accurate. \textbf{c)} Accuracy (PSNR) vs runtime of each method (averaged over n=28 true blurred images using unseen coefficients), with the
number of model parameters (written below each circle) determining the size of the circle. Noise is Poisson distributed with SNR 15. 
 }
 \label{fig:poisson}
\end{figure} 

\section{Seidel Coefficients}
\label{appendix:seidel}

\begin{figure}[h]
 \centering
\includegraphics[width=\linewidth]{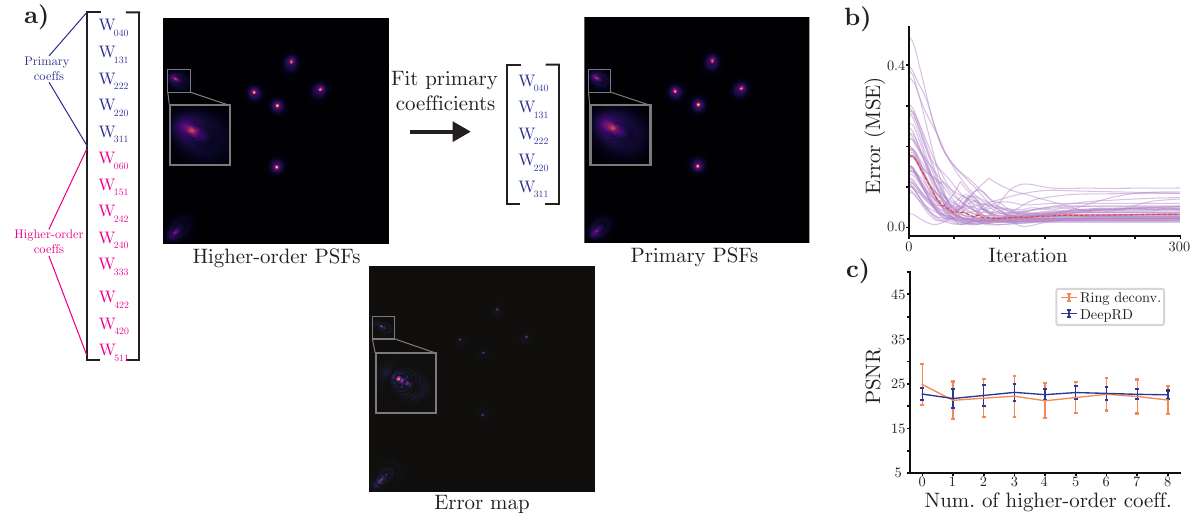}
 \caption{\textbf{Approximation of higher-order Seidel coefficients with primary Seidel coefficients.}  \textbf{a)} An example trial: a random set of 13 Seidel coefficients are selected and used to generate a noisy calibration image. Next, only 5 primary Seidel coefficients are fit to the calibration image and are used to generate estimated PSFs. An error map between the high-order and fitted PSFs is shown at the bottom; some of the high-frequency features are missed by the lower-order approximation. \textbf{b)} Error in the fit between the estimated primary coefficients and true primary coefficients in the presence of higher-order coefficients. 50 trials are plotted with the red curve indicating the median convergence. The correct primary coefficients are estimated a majority of the time. We hypothesize that the small bump in the plot where the error slightly increases before convergence is an inflation of the primary coefficients to better explain the higher-order PSFs. \textbf{c)} Images blurred with higher-order PSFs are deblurred with RDM using lower order approximations of the PSFs and the PSNRs are recorded. The experiment is repeated with increasing degrees of higher-order blur, and thus increasing discrepancy from the lower-order approximations. The average PSNR and its standard deviation over 50 trials for ring deconvolution and DeepRD are plotted as a function of the number of higher-order coefficients used in the blur, from 0 (baseline case) to 8 (all 13 coefficients). While the average PSNRs for ring deconvolution seem to drop slightly from the baseline, all orders for both methods have PSNRs within a standard deviation of the baseline error, suggesting the presence of higher-order coefficients are not enough to drastically affect RDM. For all experiments, noise was Poisson with SNR 30 and the coefficients were always normalized to unit norm. }

\label{fig:higher-order}
\end{figure} 

Here, we give a brief background on Seidel coefficients. This is not a complete treatment and we encourage an interested reader to refer to Voelz et al.\cite{voelz2011computational} for an overview, and the standard Born et al.\cite{born2013principles} for a fully rigorous derivation.

Consider a rotationally-symmetric---and consequently LRI---imaging system. It is common to consolidate the system aberrations into a single function and apply it at the exit pupil plane of the system. This complex-valued function, known as the generalized pupil function $p$ is composed of a binary-valued amplitude distribution (set by the shape of the pupil) and a phase distribution $w$, which quantifies the deviation of the pupil wavefront from the ideal spherical shape necessary for diffraction-limited imaging.
Since LRI systems are generally spatially-varying (albeit only radially), $w$ becomes a function of the distance of the source from the optical axis $r$. 
Letting $(s, t)$ be the pupil plane coordinates, we can write $p$ as
\begin{equation}
	p(s, t; r)  = circ\bigg( \frac{r}{R} \bigg) e^{w(s, t; r)}.
\end{equation}
For rotationally-symmetric systems, as in our case, it is possible to expand $w$ as an infinite power series (see the aberration function in \cite{born2013principles}). Usually only the 4th order terms of this series are used, which yields the following 2D polynomial of the pupil plane coordinates,
\begin{equation}
	w_{\omega}(s, t; r) = \omega_{s}(s^2 + t^2)^2 +  \omega_{c}(s^2 + t^2)s r +  \omega_{a} s^2 r^2 +  \omega_{f} (s^2 + t^2) r^2 +  \omega_{d} s r^3,
\end{equation}
\noindent where $ \omega = ( \omega_{s},  \omega_{c},  \omega_{a},  \omega_{f},  \omega_{d} )$ are the five primary Seidel coefficients and $w_{ \omega}$ only depends on the radial location due to the LRI assumption.
Note that while the 5 primary Seidel coefficients are a subset of the infinitely many available coefficients, they represent the most common optical aberrations: spherical, coma, astigmatism, field curvature, and distortion. In particular, these aberrations are inherent to all spherically-shaped optics. Finally, the pupil function contains the same information as the PSFs, since they are related by a Fourier transform,
\begin{equation}
	\label{eq:pupil-psf}
h(x,y; r) =  \big | \F^{-1}\{  p(-\lambda d f_{x}, -\lambda d f_{y}; r)   \}  \big |^2 ,
\end{equation}
where $\lambda$ is the wavelength of light, $d$ is the distance from the pupil plane to the image plane, and $f_x, f_y$ are the variables which the Fourier transform is taken over. Thus, knowledge of the 5 Seidel coefficients provides an approximation of the pupil function, which in turn, can accurately estimate PSFs for LRI systems, including the radial line of PSFs needed for computing ring convolution.

\subsection{Higher-order coefficients}

\label{appendix:higher-order}
The Seidel series is an infinite power series expansion and thus has an infinite number of coefficients. While the aforementioned primary coefficients are dominant in practical optics---particularly for spherical lenses---it is worth considering the impact of higher-order coefficients, especially in the context of our Seidel fitting and ring convolution/deconvolution algorithms. To that end, we include a simulated experiment that tests our methods in the additional presence of  6th order Seidel aberration coefficients. In the experiment (shown in Fig.~\ref{fig:higher-order}), a randomly-generated set of 13 (5 primary and 8 6th order) coefficients with unit norm are used to generate PSFs, which are then used to create a calibration image and blur a test image. The calibration image is noised and used to fit only the 5 primary coefficients. Finally, the fitted primary coefficients are used to deblur the blurry test image via ring deconvolution. Despite a noticeable difference in the higher-order PSFs, the fitting procedure almost always was able to converge to the correct primary coefficients, albeit with a slightly inflated values---likely accounting for the effects of the higher-order terms. The average deblurring error for ring deconvolution
drops slightly compared to the baseline case (0 high-order coefficients), but is within a standard deviation of the baseline error. DeepRD's error also remains virtually unchanged and appears even more robust to the addition of high-order coefficients. Thus even in the presence of higher-order aberrations, RDM calibrated with only the primary Seidel coefficients produces accurate deblurring. As a final note, though the primary coefficients are largely sufficient, we do additionally include the option to fit the higher-order terms in our \href{https://github.com/apsk14/lri}{codebase}.

\section{Interpretability of DeepRD}
\label{appendix:DeepRD}
One desirable property of DeepRD is its interpretability.
Because the DeepRD latent space lives in the space of Seidel coefficients, we can easily see \emph{why} our image reconstruction looks the way it does.
Moreover, if our reconstruction is deficient, we can manually adjust the latent space, thereby changing the Seidel coefficients and improving the reconstruction.
In the literature on encoder-decoder frameworks, this property is called disentanglement, that is, each element in the latent space has its own distinct and interpretable meaning.
For example, if we only increase the coefficient value for radial distortion, then we expect the network to further unwarp the predicted image correspondingly.

This disentanglement property emerges in practice. We observe this by deblurring images of insect cornea taken by the Miniscope using a sweep over each aberration coefficient independently. For each coefficient, we set all other coefficients to 0 and sweep the chosen coefficient from 0 to 3 waves. What we observe is that the network only corrects for the chosen aberration in increasing amounts, while leaving the other characteristics of the image alone. For example, when we sweep distortion, the image remains blurry, but radially warps in the opposite direction compared to if it were distorted. Results for each aberration coefficient are shown in video form \href{https://berkeley.box.com/s/11ag924u17qxe0b7d1ghe8lmfngzfofa}{here}. Such results confirm our hypothesis that DeepRD is indeed learning to perform ring deconvolution given a system's Seidel coefficients, and does not rely on knowledge of the specific image distribution. This further aids our confidence in its ability to generalize.

\section{A Fourier analysis of LRI systems}
\label{appendix:fourier}

Perhaps the main reason to make the linear space-invariant assumption is access to its Fourier space interpretation. Through an application of the Convolution Theorem to the space-invariant forward model, we see that an LSI system's output spectrum is a product of the input's spectrum with the spectrum of the impulse response, i.e., the transfer function. Thus any LSI imaging system can be thought of as a \emph{filter} which individually scales each frequency component of the input object's intensity distribution. In the context of imaging, the transfer function---called the Optical Transfer Function (OTF)---describes how the imaging system scales each spatial frequency in the sample. The OTF gives rise to valuable intuitions; for example, since all practical imaging systems have bandlimited transfer functions, the output image is a lowpass-filtered version of the sample whose maximum-achievable resolution is directly proportional to the bandlimit. Moreover system aberrations can be analyzed by comparing the OTF in the presence of aberrations with the ideal OTF.

In this section we will formally develop an analogous Fourier interpretation of ring convolution and explain how its features can similarly be used to characterize the performance of an imaging system.  While the LRI interpretation is more complicated than its LSI counterpart, it still provides a rich and realistic view of how LRI imaging systems transmit frequencies. As an example of its utility, we will see that the LRI model allows for a more general, radially-dependent notion of system resolution. For a review of the ring convolution operation please refer to Fig.~\ref{fig:lri_fwd}

\begin{figure}[t]
 \centering
\includegraphics[width=\linewidth]{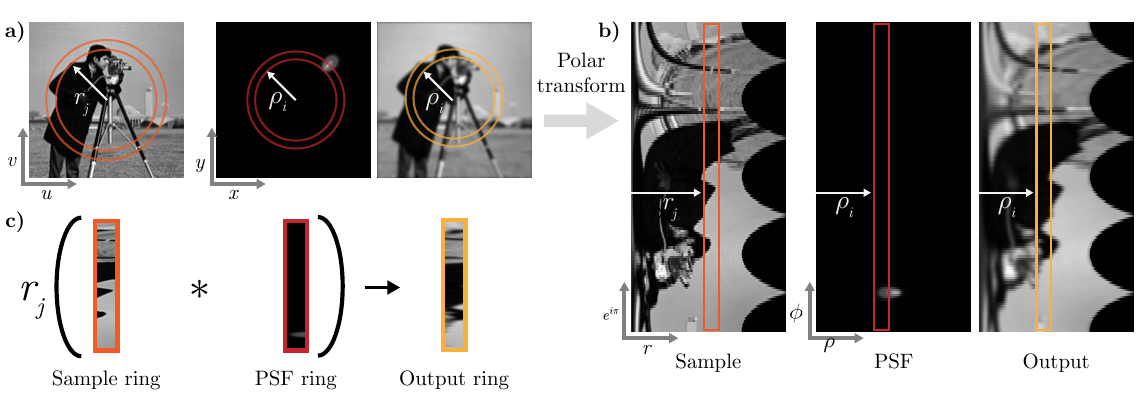}
 \caption{\textbf{Ring convolution.} A single step of the ring convolution algorithm. Here we are solving for a single concentric ring of the output denoted by its distance $\rho_i$ from the center. \textbf{a)} Left to right: sharp 'sample`, the system PSF at radius $r_j$, and blurry output of ring convolution.\textbf{b)} Corresponding polar resamplings of the images in a), this allows us to extract rings. \textbf{c)} The output ring at $\rho_i$ is the sum of 1D convolutions; the $\rho_i$ ring of the PSF at position $r_j$ is convolved with the $r_j$ ring of the sample for every radius $r_j$.
 }
 \label{fig:lri_fwd}
\end{figure} 

To begin, we will define a slightly different notion of spatial frequency, called rotational frequency and denoted $\Theta$, which is the quantity that gets filtered by an LRI system. 
\begin{definition}[Rotational Fourier Transform (RoFT)]
Let $f:\R^2 \to \C$ represent a (potentially complex-valued) image, and let $\tilde{f}$ be it's polar counterpart (as per~\ref{eq:lri}).
The RoFT of $f$ is given by
\begin{equation}
    \tilde{F}(r,\xi) = \int \tilde{f}(r, \theta) e^{i2\pi\theta \xi}d\theta.
\end{equation}
\end{definition}

\begin{figure}[t]
	\centering
	\includegraphics[width=\linewidth]{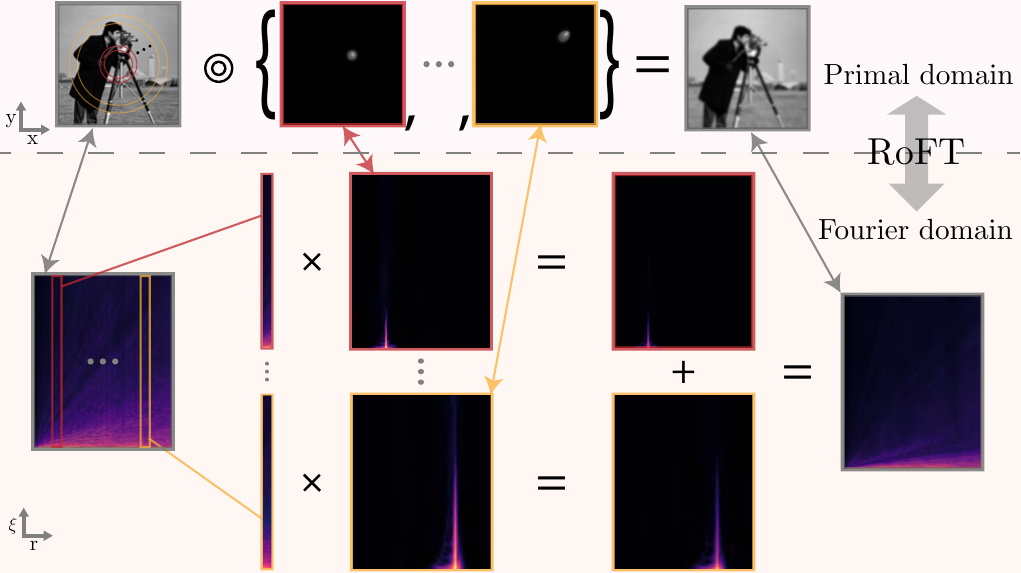}
	\caption{\textbf{LRI filtering.} The top, pink box shows LRI filtering in the spatial domain: an object (left) is LRI filtered by a radial set of PSFs (middle) to give a blurry image (right). The bottom, orange box shows the Fourier domain equivalent. Double arrows indicate Rotational Fourier Transform (RoFT) pairs. Strips at each $r$ of the object RoFT (left) are individually multiplied by the RoFTs of the corresponding PSF at $r$ (middle left) producing filtered contributions (middle right). Finally each of these contributions for every $r$ is summed to form the RoFT of the blurry image (right). }
	\label{fig:lri_fourier}
\end{figure} 
Intuitively, one can think of the values of rotational frequencies as quantifying how quickly the image can change as one travels in a ring of radius $r$ around the center. However, this notion of oscillation speed depends on the radius---think about the spokes on the wheel of a bicycle, even though the spokes are evenly spaced for any given radius, they become further spaced for larger radii. However, under polar sampling they are scaled to have the same frequency. This scaling can be seen in Fig.~\ref{fig:lri_fourier} as the RoFT of each quantity tend to have higher frequencies at larger radii.

Now, writing the Ring Convolution Theorem in terms of the RoFTs of each quantity (denoted with capital letter and tilde), we see that 
\begin{equation}
    \tilde{F}(\rho,\xi) = \int r \tilde{G}(r,\xi) \tilde{H}(\rho, \xi; r)dr.
\end{equation}
Like its LSI counterpart, the above model can be thought of a filtering operation, but now a more complex one; the object's values $r$ away from the center are \emph{filtered} and \emph{mixed} by the $r^{th}$ PSF. To more fully understand this interpretation, consider the object's values at some radius $r$; they form a ring which will be filtered by the $r^{th}$ PSF. Specifically, the spectrum of this ring, or equivalently its RoFT at $r$, is point-wise multiplied by each ring in the RoFT of the $r^{th}$ PSF, which yields a set of filtered rings indexed by $\rho$. The $\rho^{th}$ filtered ring represents the contribution of the object ring at $r$ to the image ring at $\rho$ (see Fig.~\ref{fig:lri_fourier}). 

\begin{figure}[t]
	\centering
	\includegraphics[width=0.75\linewidth]{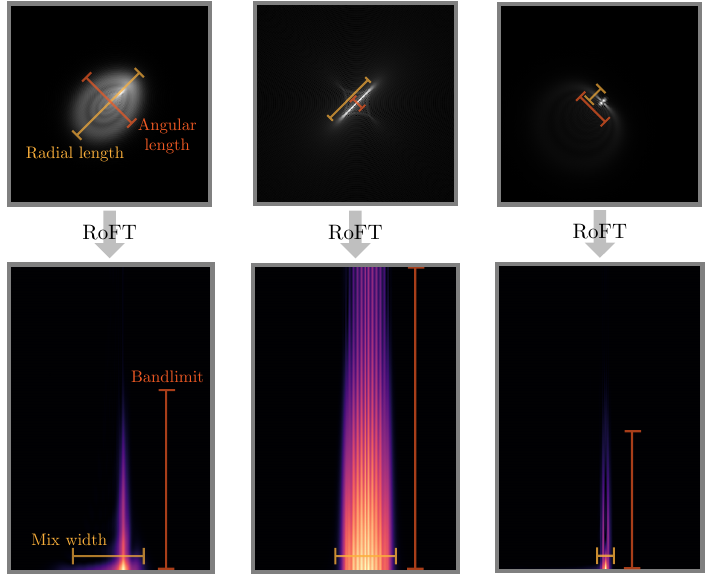}
	\caption{\textbf{LRI PSF interpretation.} Evaluating an LRI optical system amounts to considering the length of each PSF in the angular and radial directions. These lengths are inversely proportional to the bandwidth and mix width of the corresponding RoFT, respectively (see leftmost column). In the middle column we see an elogated PSF in the radial direction (i.e., astigmatism) which has a large bandwidth, but incurs more off-radius mixing due to a large mix width. The rightmost column shows an elongated PSF in the angular direction which has a small mix width, but has a relatively small bandwidth.}
	\label{fig:lri_psf_interp}
	
\end{figure} 

It follows that the shape of the $r^{th}$ PSF determines both how much the object's rotational frequencies are filtered through the system and which parts of the object are mixed together to form the image. By looking at the shape of the $r^{th}$ PSF we can determine the exact nature of this filtering and mixing. As shown in Fig.~\ref{fig:lri_psf_interp} the angular length of the PSF determines the height of the PSF's RoFT at $\rho$, which tells us which of the object rotational frequencies at $r$ will make it to the image at $\rho$---this is analogous to the OTF bandwidth in the LSI case. Meanwhile, the radial extent of the PSF (i.e., how many concentric rings it covers) controls the extent of the image rings (i.e, which $\rho$ values) are effected by the object at $r$. This manifests as the width of the PSF's RoFT which we call the \emph{mix width}. 

Now our notion of resolution, assuming the ability to perfectly de-mix, is radially dependent and, at radius $r$, is proportional to the angular arc length at $\rho=r$ of the $r^{th}$ PSF. The LRI filtering interpretation offers a more realistic understanding of how an imaging system filters an object and opens the door for a host of new imaging techniques which optimize for key features, such as resolution, under the LRI model. An analogous analysis for systems with other forms of symmetry--- like the linear symmetry present in light-sheet microscopy---will lead to similar conclusions.

\section{Additional Related Work}

The idea of linear space-invariance (LSI)---or rather its 1D predecessor time invariance---is difficult to trace historically, but the first mention of it that we were able to find was by Richard Hamming in 1934~\cite{hamming}, followed by the far more popular definition provided in~\cite{oppenheim1975digital}. 
The concept also has a close connection to the convolution integral---which had existed separately as far back as the 1800s~\cite{dominguez_2015}---through the famous Convolution Theorem.
Much of the foundational mathematics of convolution and its frequency domain interpretation comes from the work of Laplace and Fourier; these techniques enabled the application of linear time (space)-invariant systems to communications and, of course, imaging. 
The practical use of LSI systems for imaging analysis arose from Bell labs during the emergence of digital signal processing for communication in the 1960s and 1970s.

Image deblurring has its roots much later; it was only in 1949 when Norbert Wiener published the Wiener Filter for linear time-invariant deconvolution~\cite{wiener1964extrapolation}. While the original technique was developed in 1940 for time series data in order to track enemy planes in World War II, it was soon adapted for image deblurring. With the rise of digital computing in the following decades came an explosion of image deconvolution techniques which are commonly employed today such as Richardson-Lucy deconvolution~\cite{Richardson, Lucy}. It is interesting to note that some of these algorithms have existed for nearly a century and yet there is no clear ``best'' method.

\end{document}